\documentclass{iopart}
\usepackage{hyperref}
\usepackage[english]{babel}
\usepackage[utf8x]{inputenc}
\usepackage[T1]{fontenc}
\usepackage{amssymb,amsthm}
\usepackage[pdftex]{graphicx}

\newtheorem{theorem}{Theorem}
\newtheorem{corollary}{Corollary}
\newcommand{\la}{\left\langle}
\newcommand{\ra}{\right\rangle}

\begin{document}
\title[Exact construction and sampling of graphs with prescribed degree correlations]
{Exact sampling of graphs with prescribed degree correlations}

\author{Kevin E. Bassler$^{1,2,9}$, Charo I. Del~Genio$^{3,4,5,9}$, 
Péter L. Erdős$^{6}$, István Miklós$^{6,7}$, Zoltán Toroczkai$^{8,9}$}
\address{$1$ Department of Physics, 617 Science \& Research 1, 
University of Houston, Houston, Texas 77204-5005, USA}
\address{$2$ Texas Center for Superconductivity, 202 Houston 
Science Center, University of Houston, Houston, Texas 77204-5002, USA}
\address{$3$ Warwick Mathematics Institute, University of Warwick, 
Gibbet Hill Road, Coventry CV4 7AL, United Kingdom}
\address{$4$ Centre for Complexity Science, University of Warwick, 
Gibbet Hill Road, Coventry CV4 7AL, United Kingdom}
\address{$5$ Warwick Infectious Disease Epidemiology Research 
(WIDER) Centre, University of Warwick, Gibbet Hill Road, Coventry CV4 7AL, United Kingdom}
\address{$6$ MTA Alfréd Rényi Institute of Mathematics, Reáltanoda u.\ 13--15, Budapest 1053, Hungary}
\address{$7$ Institute for Computer Science and Control, 1111 Budapest 1518, L\'agym\'anyosi \'ut 11, Hungary}
\address{$8$ Department of Physics, and Interdisciplinary Center for 
Network Science and Applications (ICeNSA) 225 Nieuwland Science Hall, 
University of Notre Dame, Notre Dame, Indiana 46556, USA}
\address{$9$ Max Planck Institute for the Physics of Complex Systems, 
Nöthnitzer Str. 38, D-01187 Dresden, Germany}

\begin{abstract}
Many real-world networks exhibit correlations
between the node degrees. For instance, in social
networks nodes tend to connect to nodes of similar
degree and conversely, in biological and technological
networks, high-degree nodes tend to be linked
with low-degree nodes. Degree correlations also
affect the dynamics of processes supported by
a network structure, such as the spread of opinions
or epidemics. The proper modelling of these systems, i.e., without 
uncontrolled biases, requires the sampling of networks with a specified
set of constraints. We present a solution to
the sampling problem when the constraints imposed
are the degree correlations. In particular, we
develop an exact method to construct
and sample graphs with a specified joint-degree
matrix, which is a matrix providing the
number of edges between all the sets of nodes
of a given degree, for all degrees, thus completely 
specifying all pairwise degree correlations, and additionally, the
degree sequence itself. Our algorithm
always produces independent samples without backtracking.
The complexity of the graph construction algorithm is 
${\cal O}(NM)$ where $N$ is the number of nodes and $M$ is the number of 
edges.
\end{abstract}
\submitto{\NJP}
\maketitle

\section{Introduction}
Complex systems often consist of a discrete set of elements with 
heterogeneous pairwise interactions. Networks, or graphs
have proven to be a useful representational paradigm for the study of these
systems~\cite{New03,Tor04,Boc06,Boc14}. The nodes, or vertices,
of the graphs represent the discrete elements, and the edges,
or links, represent their interaction. In empirical studies
of real-world systems, however, for reasons of methodology,
privacy, or simply lack of data, frequently there is only limited information available 
about the connectivity structure of a network.
When this is the case, one has to take a statistical approach and study 
ensembles of graphs that conform to some structural constraints. 
This statistical approach enables the computation of ensemble averages of network
observables as determined solely by the constraints, i.e., by the specified structural
properties of the graphs. Ensemble modeling of this type is necessary
to determine the relationship between the given structural constraints and
the behavior of the complex system as a whole.
Calculating ensemble averages, though,
requires the ability to construct all the graphs that are consistent
with the required structural constraints, a highly non-trivial problem.

Perhaps one of the simplest examples of structural constraints that occur
in data-driven studies of real-world systems is to fix the \emph{degree}
of each node, which is the number of edges that are connected
to, or are incident on the node. For an undirected graph with $N$ nodes this
information is specified by a \emph{degree sequence} ${\mathcal D}=
\left\lbrace d_1,d_2,\cdots,d_N\right\rbrace$,
where $d_i$ is the degree of node $i$. Similarly, for a directed
graph, a \emph{bi-degree sequence} (BDS)
${\mathcal D}=\left\lbrace \left(d_1^{-},d_1^{+}\right), 
\left(d_2^{-},d_2^{+}\right), \cdots, \left(d_N^{-},d_N^{+}\right)\right\rbrace$
specifies the number of incoming and outgoing edges for each node where
 $d_i^{-}$ denotes the in-degree, and $d_i^{+}$ the out-degree,
of node $i$. The situation of most practical interest is when we demand
the graph with a given degree sequence to be a simple graph, which has the
additional constraints that there can be at most one link (in each direction, if
directed) between any two nodes, and that no link starts and ends on the same node (no self-loops).
However, not all positive integer sequences 
can serve as the sequence of the degrees of some simple graph.
If such a graph does exist, then
the sequence is said to be \emph{graphical}. Any simple graph
(just ``graph'' from here on) with the
prescribed node degrees is said to \emph{realize} the degree sequence,
and it is called a \emph{graphical} realization of the sequence.
The two main results used to test the graphicality of an undirected degree sequence are the
Erdős-Gallai theorem~\cite{Erd60} and the Havel-Hakimi theorem~\cite{Hav55,Hak62}.
For directed networks, instead, the main theorem characterizing
the graphicality of a BDS is due to Fulkerson~\cite{Ful60}. More
recently, exploiting a formulation based on recurrence relations,
new methods were introduced to implement these tests with a worst
case computational complexity that is only linear in the number
of nodes~\cite{Del10,Kir11,Kim12}. The advantage of these methods over
others with similar complexity~\cite{Hel09} is that they also allow
a straightforward algorithmic implementation.

While the above results provide complete and practical answers
to the question of the graphicality of sequences of integers,
they do not suffice to solve the problem of constructing graphs
with prescribed degrees. One of the main issues with constructing
graphs for the purpose of ensemble modeling is that, except for
networks of just a few nodes, the number of graphs realizing a
degree sequence, or other possible constraints, is generally so
large that their complete enumeration is impractical. Therefore,
one has to resort to \emph{sampling} the space of realizations
by randomly generating networks with prescribed node degrees~\cite{Del10,Kim12}.
For the case of degree-based graph sampling, the existing approaches
generally fall into two classes that can broadly be referred
to as ``rewiring'' and ``stub-matching''. Rewiring methods start from
a graph with the required degrees and use Markov chain Monte~Carlo (MCMC)
schemes to swap repeatedly the ends of pairs of edges to produce
new graphs with the same degree sequence~\cite{Tay82,Rao96,Kan99,Vig05}.
Stub-matching methods, instead, are direct construction algorithms
that build the graphs by sequentially creating the edges 
via the joining of two stubs of two nodes~\cite{New01,Bog04,Cat05,Ser05,Bri06}. A stub represents
a non-connected, ``dangling half-edge'' and a node has as many stubs as its degree.
Unfortunately, these techniques can provide biased results, or
are ill-controlled. In the case of the MCMC method the mixing time is in general unknown 
and thus one cannot know \textit{a priori} the number of swaps needed to produce two 
statistically independent samples. Proofs showing polynomial mixing of the MCMC method
have been recently developed for special degree sequences ~\cite{Coo07,Gre11,Mik13,Erd13},
and for the case of balanced realizations of joint-degree matrices~\cite{Erd15}. However,
none of these methods allows the determination of the exponent of the polynomial scaling.

Among the stub-matching methods, the most commonly used
algorithm, which is also ill controlled, is known as the configuration model.
The configuration model was proposed in~\cite{New01} as an algorithmic
equivalent of the results from Refs.~\cite{Mol95,Mol98},
themselves based on prior models~\cite{Ben78,Bol80}.
The algorithm randomly extracts two stubs from the set of all stubs 
not yet connected into edges, and connects them into an edge. 
If a multi-edge or a self-loop has just been created, the process is 
restarted from the very beginning to avoid biases. However, depending on the degree sequence,
this process can become
very inefficient with an uncontrolled running
time, just like the MCMC method.
Alternatively, one can ignore multi-edges and self-loops,
and fix them ``by hand'' at the end of the process. However,
doing so produces
significant biases even in the limit of large system size~\cite{Kle12}.
Recently, a novel family of stub-matching algorithms 
were introduced for both undirected~\cite{Del10} and directed~\cite{Kim12}
degree sequences (reproduced here in  \ref{apxA}), based on the so-called star-constrained
graphicality theorems~\cite{Kim09,Erd10}. These algorithms generate statistically
independent samples with a worst case polynomial time of $O(NM)$, where $M$ is the
total number of edges. The samples are not generated uniformly. However, their statistical
weights are computable and can be used
to obtain results in an importance sampling framework~\cite{Del10,Bli10,Kim12,Kor13}.
Note that the solution for the directed
sequences also solves the
problem for bipartite sequences because a
bipartite graph can always be represented as a directed one in
which one of the two sets of nodes has only outgoing edges, and
the other set has only incoming ones.

Graph construction and sampling becomes even more difficult
when there are structural constraints of higher order, such
as correlations amongst the node degrees. Degree correlations
can be expressed in several ways, for example with the help of the conditional
probability $P(d'|d)$ that a node of degree $d$ will have a neighbor
of degree $d'$, or more simply, by the average degree of the neighbors of
a node with degree $d$, $\bar{d'}(d) = \sum_{d'} d' P(d'|d)$
\cite{Ves01}.
The properties of $\bar{d'}(d)$ characterize the
so-called \emph{assortativity} of a graph, which is a measure
of the tendency of a node to connect to nodes of similar
degree. If $\bar{d'}(d)$ is increasing in $d$, the
graph is degree assortative, if it is decreasing the graph is degree disassortative,
and if it is constant, the graph is degree uncorrelated. Even more
coarse-grained measures of degree correlations are possible,
including the Pearson coefficient~\cite{Pea895}, the Spearman
coefficient~\cite{Yul10} and the Kendall coefficient~\cite{Ken38}.
These coefficients assume values ranging from~$-1$, for highly
disassortative graphs, to~1, for highly assortative ones. 

A more precise way to express degree correlations
is via the use of a joint-degree matrix.
The \emph{joint-degree matrix}
(JDM) of a given undirected simple graph
is a symmetric matrix whose $\left(\alpha,\beta\right)$ element
is the number of edges between nodes of degree $\alpha$ and
nodes of degree $\beta$. The dimensions of the JDM
are $\Delta\times\Delta$,
where $\Delta$ is the largest degree of a node in the graph.
The degree correlation measures
discussed above specify the correlations only statistically, but they do not 
fix the number of edges between nodes of given degrees, whereas
the joint-degree matrices do. In this sense, the relationship
between joint-degree matrices and the statistical degree correlation measures 
is similar to the relationship between degree sequences and degree distributions.

Degree correlations have generated
considerable interest, as they are known to affect many structural
and dynamical properties of graphs and the processes they support~\cite{New02,Mas02,New03_2,Eub04,Dag12,Del13,You13,Wil14}.
Nevertheless, even though their importance is well established,
it has heretofore not been possible to perform ensemble modeling
of graphs with prescribed joint-degree matrices. In this Article,
we solve this problem by developing an algorithm based on the
stub-matching method to construct and sample ensembles of graphs 
with a specified joint-degree matrix.

\section{Mathematical foundations}

\subsection{Graphicality of JDMs}

The problem of graphicality for JDMs asks whether a specified
symmetric matrix can be the JDM of a simple graph. Our starting
point is an Erdős-Gallai-like theorem that gives the requiements
for a JDM to be graphical~\cite{Pat76,Sta12,Cza13}.

Before stating the theorem, though, note that a JDM specifies
uniquely the degree sequence of the graphs that realize it~\cite{Pat76}.
Given a JDM $J$, the number of nodes with degree $\alpha$ is
\begin{equation*}
 \left|V_\alpha\right| = \frac{1}{\alpha}\left(J_{\alpha\alpha}+
 \sum_{\beta=1}^\Delta J_{\alpha\beta}\right)\:,
\end{equation*}
where $V_\alpha$ is the set of nodes, or \emph{degree class},
with degree $\alpha$. As a general rule of notation we will use lowercase Greek letters
to indicate degree values and lowercase Latin letters for node indices.
In the equation above the sum of each row $\alpha$ of $J$ is the number 
of connections involving nodes of degree $\alpha$ (i.e., all nodes in
class $V_{\alpha}$). As each node of degree $\alpha$
has exactly $\alpha$ stubs the total number
of nodes of degree $\alpha$ is given by the notal number of stubs from all nodes
in class $V_{\alpha}$ divided by $\alpha$.
Moreover, each edge between nodes \emph{of the same degree} involves~2
stubs. Thus, the diagonal elements must be double-counted. Note that
multiple JDMs can specify the same degree sequence and thus
prescribing a JDM is more constraining than only prescribing
a degree sequence.
With the definitions above, the necessary and sufficient conditions 
for a JDM to be graphical can be stated as follows~\cite{Pat76,Sta12,Cza13}:
\begin{theorem}[JDM graphicality] \label{thm1}
A symmetric $\Delta\times\Delta$ matrix $J$ with non-negative integer elements is a graphical JDM if and only if:
\begin{eqnarray*}
1) \quad \left|V_\alpha\right|\mathrm{\ is\ an\ integer\ }&\forall 1\leqslant\alpha\leqslant\Delta,\\
2)\quad J_{\alpha\alpha}\leqslant{{\left|V_\alpha\right|} \choose 2}&\forall 1
\leqslant\alpha\leqslant\Delta, \; \mbox{and}\\
3) \quad J_{\alpha\beta}\leqslant\left|V_\alpha\right|\left|V_\beta\right|&\forall 1
\leqslant\alpha,\beta\leqslant\Delta\mathrm{\ and\ }\alpha\neq\beta\:.
\end{eqnarray*}
\end{theorem}

It is important to observe 
that any graphical realization of a JDM can be decomposed into the disjoint union of a set of 
subgraphs $G_{\alpha\beta}$ that are bipartite ($\alpha \neq \beta$) with node sets $V_{\alpha}$ 
and $V_{\beta}$ and $J_{\alpha\beta}$ edges between them or unipartite ($\alpha = \beta$)
with node set $V_{\alpha}$ and $J_{\alpha\alpha}$ edges within that set. 
We are going to call such representation of a graphical realization a degree class representation. 
\begin{figure}
\centering
\includegraphics[width=1.0\textwidth]{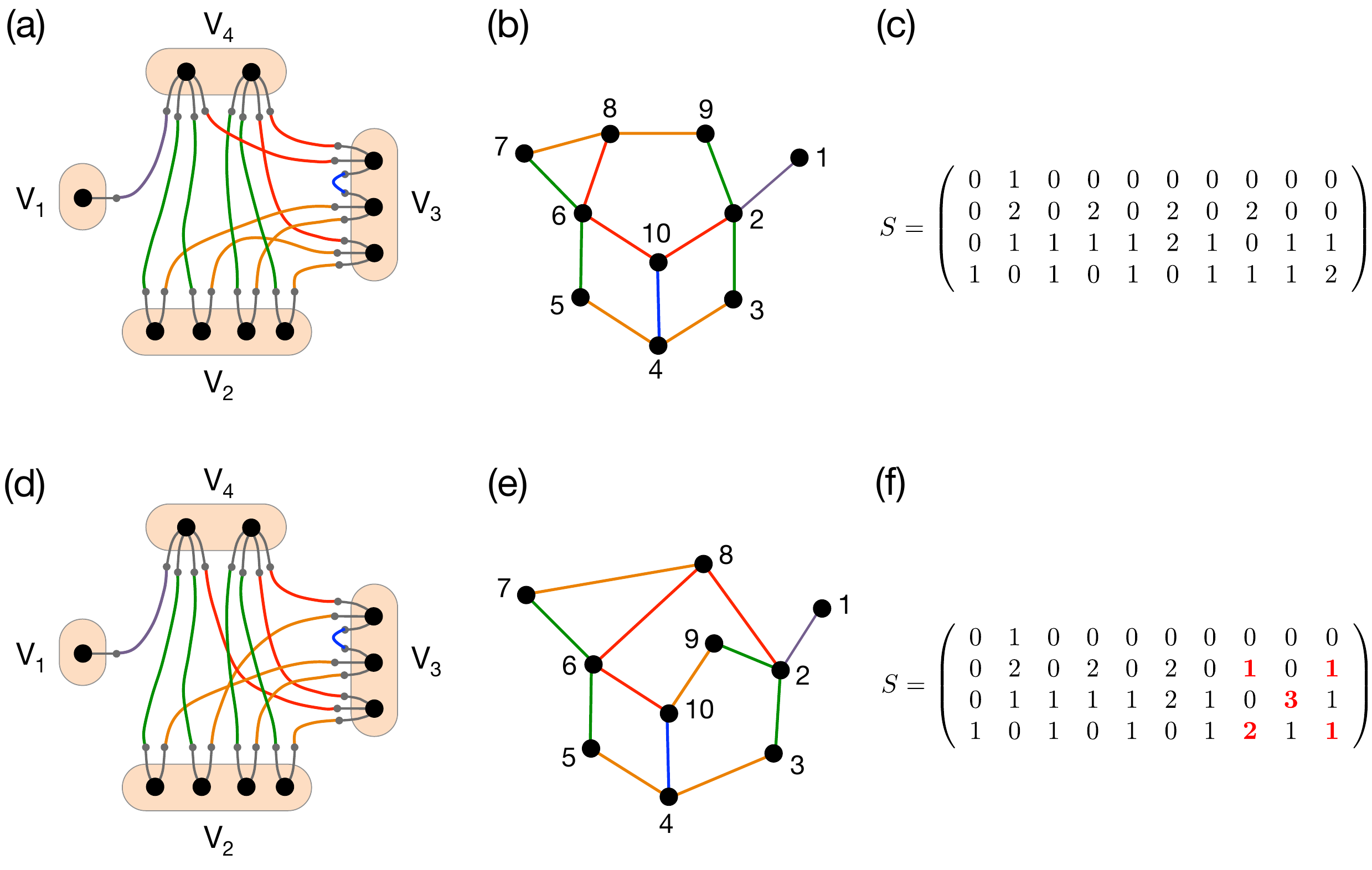}
\caption{\label{fig:ex1} Graphical realizations of a simple JDM, given in (\ref{ex1}).
Panels (a) and (d) are degree class representations, while panels (b) and (e) are regular
representations. The color of the edges indicates the subgraph $G_{\alpha\beta}$ they
belong to. Panels (c) and (f) show the corresponding degree-spectra matrices for the two
realizations; they differ in the bold red entries.}
\end{figure}
A simple example of a graphical JDM with $N=10$ and $\Delta = 4$ is given by the matrix:
\begin{eqnarray}
J = 
\left(\begin{array}{cccc}
0 & 0 & 0 & 1 \\
0 & 0 & 4 & 4 \\
0 & 4 & 1 & 3 \\
1 & 4 & 3 & 0
\end{array}\right)\;. \label{ex1}
\end{eqnarray}
Panels~(a) and~(b) of Fig.~\ref{fig:ex1} show a graphical realization of $J$ in degree class representation
and regular representation, respectively. Panels~(d) and~(e) of the same figure show another realization
of $J$ in the two representations. The color of the edges indicate the subgraph they belong to.
For example, $G_{24}$ is a bipartite graph between nodes of degree 2 ($V_2$)  and 4 ($V_4$),
respectively,  having $J_{2,4} = 4$ edges drawn in green color, whereas $G_{33}$ is unipartite
with a single $J_{33}=1$ edge drawn in blue. Note that while both graphical realizations have the
same JDM, they are very different graphs. To see this, consider the counts $n_{\ell}$
of cycles $C_{\ell}$ of length $\ell$ (a cycle is a closed path without repeated nodes). The graph in
Fig.~\ref{fig:ex1}(b) has $n_3 = 1$, $n_4 = 2$, $n_5 = 1$, $n_6 = 2$, $n_7 = 3$ and $n_8 = 3$,
whereas the one in Fig.~\ref{fig:ex1}(e) has  $n_3 = 1$, $n_4 = 1$, $n_5 = 2$, $n_6 = 3$, $n_7 = 4$
and $n_8 = 1$.

Theorem~\ref{thm1} is an existence theorem, just like the Erd\H{o}s-Gallai
theorem for the case of degree sequences, and as such it does not provide an 
algorithm that can generate simple graphs with a given JDM. More importantly, we 
also need an algorithm that does not exclude classes of graphical realizations of a given 
JDM, but that can construct in principle any such realization.
The situation is similar to that of degree sequences. In that case the Havel-Hakimi
method~\cite{Hav55,Hak62} is always able to create a graphical realization of a graphical
degree sequence, but cannot construct them all, i.e., there will be some realizations
that can never be built by this algorithm. This was the reason for the introduction of the
notion of star-constrained graphicality in Refs.~\cite{Kim09,Erd10} and the subsequent construction
algorithms in Refs.~\cite{Del10,Kim12}. Here as well, we want to have a direct construction algorithm
and ultimately an exact sampler that does not exclude any realization of a JDM. 
Due to the different nature of the constraints from the degree-sequence-based case, we need to develop a novel approach.

The idea of the approach is based on the degree class representation above. Since the edges 
of the subgraphs $G_{\alpha\beta}$ are disjoint, we could build a graphical realization $G$ of the JDM 
$J$ by building all these subgraphs, while respecting the constraints. For a $G_{\alpha\beta}$ subgraph
we know its node set(s) and its total number of edges $J_{\alpha\beta}$. Consider then a node
$v \in V_{\alpha}$. We are not given its degree in $G_{\alpha\beta}$ for any $\beta$, but we know that the sum of
its degrees within every one of these subgraphs must add up to $\alpha$. For example, the sum of the 
numbers of the purple, green and red edges coming out of node 2 in Fig.~\ref{fig:ex1}(b) must add to 4.
In addition, we also have the constraints that the sum of the degrees of one color of all nodes within 
$V_{\alpha}$ must equal to the corresponding given JDM entry. Indeed, for example, the sum of
all green edges in Fig.~\ref{fig:ex1}(a) or Fig.~\ref{fig:ex1}(b) is $J_{2,4} = 4$,
for orange is 4, red is 3, etc. Thus, the idea of the algorithm is to first determine
the degree of a given color respecting the constraints for all nodes and all
colors, then use our methods introduced earlier \cite{Del10,Kim12} (see \ref{apxA}) to build
the $G_{\alpha\beta}$ subgraphs based on the corresponding degree sequences of
their nodes. Different graphical realizations will be obtained from different
assignments of color degrees and, of course, from the different graphical realizations
of the same set of degrees. Note that for the bipartite subgraphs $G_{\alpha\beta}$ we are
specifying degree sequences for nodes in both partitions $V_{\alpha}$ and $V_{\beta}$
and thus we can use our graph construction method for directed graphs \cite{Kim12},
because a bipartite graph can be represented as a directed graph if nodes in one partition
have only outgoing edges and  in the other only incoming edges. In the
following it will be useful to introduce the notion of degree spectra,
representing the degrees of different colors of a node, as described above.

\subsection{Degree spectra}\label{degspe}

Consider a single row $\alpha$ of a graphical
JDM $J$. The information contained in the row
determines the precise number of edges needed
between nodes of degree $\alpha$ and nodes of
every degree. In other words, of all the stubs
coming from $V_{\alpha}$,
$J_{\alpha,1}$ of them must end in a node of degree~1,
$J_{\alpha,2}$ of them must end in a node of degree~2,
and so on. However, these matrix elements do
not specify how to distribute these edges within and between
the degree classes. To better specify these connections
one introduces the notion of the \emph{degree spectra}, which 
can be conveniently
represented as a matrix. The \textit{degree spectrum} of a node
is the sequence of its degrees
towards all the degree classes, including its own degree class. 
A \emph{degree-spectra matrix}
$S$ is a $\Delta\times N$ matrix whose $\left(\alpha,i\right)$
element $S_{\alpha i}$ is the number of edges between node
$i$ and degree class $\alpha$ (the set of nodes of degree $\alpha$). The $i^\mathrm{th}$
column of $S$ defines the degree spectrum
of node $i$. Panels~(c) and~(f) of Fig.~\ref{fig:ex1}show two representations of the
same JDM given in Eq.~\ref{ex1}. In general, there are many degree
spectra matrices that correspond to the same
JDM. As described in the previous
section, we employ a two-step
process in order to randomly sample graphs that
realize a given JDM. First, we generate a random degree-spectra
matrix from the JDM. Second, we construct a
random graph that realizes the JDM and that
obeys or is consistent with the chosen spectra matrix. This approach
creates the need for a method to guarantee
that the spectra generated from a JDM are graphical.
\begin{figure}
\centering
\includegraphics[width=1.0\textwidth]{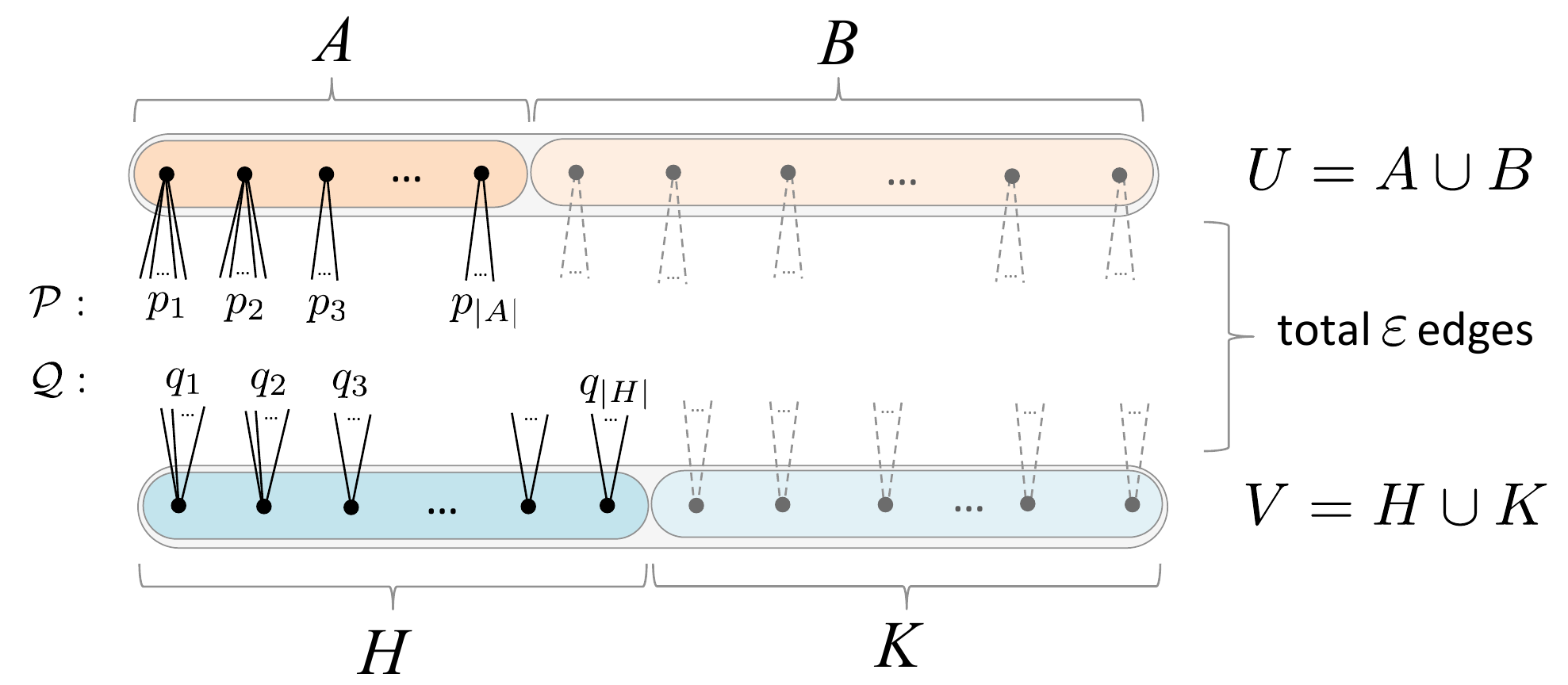}
\caption{\label{fig:schematic}Schematic for the partial degree sequence problem.}
\end{figure}

The generation of a graphical degree-spectra matrix proceeds
systematically, node by node. Therefore, at each step, some
nodes will have an already fixed number of links within some
of the subgraphs (links of a given color), while for the rest
these numbers will not have been determined yet. Thus, at any
time during this process  we have a partial degree sequence of
a bipartite graph.  As the subgraphs must be simple graphs (realizable),
one must be able to  decide whether a partial bipartite degree
sequence is graphical. The sufficient and necessary criterion
for the graphicality of a partial bipartite degree sequence will
be given in Theorem \ref{gratri} below. However, that will not
necessarily mean that the whole JDM $J$ is still realizable,
in other words, how do we know that by guaranteeing the graphicality for a subgraph
$G_{\alpha\beta}$ we have not precluded graphicality of some other 
subgraph $G_{\gamma\delta}$, and ultimately of $J$? The answer to this question
will be given by Theorem~\ref{decomp}, later on. Together, these theorems form the basis for 
our algorithm to generate graphical degree spectra.

Before proving a theorem that provides a graphicality
test for partial bipartite degree sequences, we need to set some notations. Let $A$, $B$, $H$ and
$K$ be four sets of nodes:
\begin{eqnarray*}
 A = \left\lbrace a_1, a_2, \cdots, a_{\left|A\right|}\right\rbrace & \ \ B = \left\lbrace b_1, b_2, \cdots, b_{\left|B\right|}\right\rbrace & \quad \mathrm{with\ }A\cap B = \emptyset\\
 H = \left\lbrace h_1, h_2, \cdots, h_{\left|H\right|}\right\rbrace & \ \ K = \left\lbrace k_1, k_2, \cdots, k_{\left|K\right|}\right\rbrace & \quad \mathrm{with\ }H\cap K = \emptyset
\end{eqnarray*}
and let $U=A\cup B$ and $V=H\cup K$ (see Fig.~\ref{fig:schematic}). The sets
can be of different size, but neither $U$ nor $V$
can be empty. Now, let $\mathcal{P}=\left\lbrace p_1, p_2, \cdots, p_{\left|A\right|}\right\rbrace$
and $\mathcal{Q}=\left\lbrace q_1, q_2, \cdots, q_{\left|H\right|}\right\rbrace$
be two given sequences of integers. They will represent the partial bipartite degree sequences that 
have already been fixed by the algorithm up to that point.
The degrees of the other nodes, specifically those in the sets $B$ and $K$, are not yet specified. 
What is specified is the total number of edges $\varepsilon$ in the bipartition, i.e., the total number 
edges running between the sets $U$ and $V$. Then, the partial
bipartite degree sequence triplet $\left(\mathcal{P},\mathcal{Q},\varepsilon\right)$,
hereafter simply called a \emph{triplet}, is graphical
if there exists a bipartite graph on $U$ and $V$
with $\varepsilon$ edges and degree sequences $\mathcal{D}(U)\big|_{A} = \mathcal{P}$ and 
$\mathcal{D}(V)\big|_{H}=\mathcal{Q}$.
In other words, the bipartite graph must be such that the nodes in
$A$ have degree sequence $\mathcal{P}$
and those in $H$ have degree sequence
$\mathcal{Q}$. The partial degree sequence problem is to decide whether one can choose the
degrees of the nodes in the sets $B$ and $K$ such that the above constraints are satisfied and the bipartite
degree sequence $\cal D$ is graphical.

Since the graph realizing a triplet is bipartite,
the number of edges $\varepsilon$ equals
the number of stubs in either set of nodes:
\begin{equation*}
 \varepsilon = \sum_{i=1}^{\left|U\right|}d_{u_i} = \sum_{i=1}^{\left|V\right|}d_{v_i}\:.
\end{equation*}
The imposed partial
sequences $\mathcal{P}$ and $\mathcal{Q}$
prescribe a certain number of
stubs in the first $\left|A\right|$ nodes
of $U$ and in the first $\left|H\right|$
nodes of $V$. Let these be $P=\sum_{i=1}^{\left|A\right|}p_i$
and $Q=\sum_{i=1}^{\left|H\right|}q_i$,
respectively. Then, the set $B$ must contain
exactly $\varepsilon-P$ stubs; similarly,
the set $K$ must contain exactly $\varepsilon-Q$
stubs. With these considerations, we first define
the concept of a \emph{balanced realization} of a triplet.
Let $\mu\equiv\frac{\varepsilon-P}{\left|B\right|}$
and $\nu\equiv\frac{\varepsilon-Q}{\left|K\right|}$.
A realization of a triplet is defined to be balanced if and only if the
degree of any node in $B$ is either $\left\lfloor\mu\right\rfloor$
or $\left\lceil\mu\right\rceil$, and the degree of
any node in $K$ is either $\left\lfloor\nu\right\rfloor$
or $\left\lceil\nu\right\rceil$. Notice that this means
that if $\mu$ or $\nu$ are integers, then all the nodes
in $B$ or $K$ must have exactly degree $\mu$ or $\nu$, respectively.
Conversely, if they are not integers, then the degrees
of any two nodes in $B$ or in $K$, respectively, can differ at most
by~1. That is, a realization is balanced
if and only if all the degrees of the nodes that one
is free to choose (those in $B$ and $K$) are as close as possible to their
averages $\mu$ and $\nu$. The definition can be equivalently formalized by 
introducing a functional $f$ acting
on $B$ and $K$:
\begin{eqnarray*}
 f\left(B\right) &\equiv \sum_{i=1}^{\left|B\right|}\left\lfloor \left|d_{b_i}-\mu\right|\right\rfloor\;\;\;\mbox{and}\;\;\;
 f\left(K\right) &\equiv \sum_{i=1}^{\left|K\right|}\left\lfloor \left|d_{k_i}-\nu\right|\right\rfloor\:.
\end{eqnarray*}
Then, a realization of a triplet is balanced if and only if both $f\left(B\right)$
and $f\left(K\right)$ vanish.

An important
theorem about the graphicality of triplets
can now be proven.
\begin{theorem}\label{gratri}
The triplet $\left(\mathcal{P},\mathcal{Q},\varepsilon\right)$
is graphical if and only if it admits a balanced realization.
\end{theorem}
\begin{proof}
Sufficiency is obvious. If the triplet admits any realization,
balanced or not, it is graphical by definition.

To prove necessity, suppose the triplet
is graphical. Then, it admits a realization $G$. If
$G$ is balanced, then there is nothing to do.
Conversely, if $G$ is not balanced, then
$f\left(B\right)$, $f\left(K\right)$, or both,
are greater than~0. Without loss of generality,
assume that $f\left(B\right)>0$.
Then, there exists a node $b_i \in B$ such
that either $d_{b_i}<\left\lfloor\mu\right\rfloor$
or $d_{b_i}>\left\lceil\mu\right\rceil$.
Again without loss of generality,
assume that $d_{b_i}<\left\lfloor\mu\right\rfloor$
(the other cases are treated analogously).
Then, since the number
of stubs within $B$ is fixed, there must
exist a node $b_j \in B$ such that $d_{b_j}>\left\lfloor\mu\right\rfloor$
and thus $d_{b_j}>d_{b_i}$.
But then, there must exist a node $v_k\in V$
such that $v_k$ is connected
to $b_j$ but not to $b_i$.
Now, remove the edge $(v_k,b_j)$
and replace it with $(v_k,b_i)$.
This yields a different realization with the same
degrees for the nodes in $V$,
and in which $f\left(B\right)$ is decreased by at least~1,
as the degrees of $B$ moved towards the
balanced condition. The procedure can be repeated
until $f\left(B\right)=0$, resulting
in a balanced realization.
\end{proof}
A key consequence of this theorem is the following.
\begin{corollary}\label{cortri}
Let $\left(\mathcal{P},\mathcal{Q},\varepsilon\right)$
be a graphical triplet, and let $x$ be a node
in $B$ or in $K$. If there is a realization of
the triplet in which $d_x=\alpha$ and another in
which $d_x=\beta$, with $\alpha<\beta$, then for all
$\gamma$ with $\alpha\leqslant\gamma\leqslant\beta$
there exists a realization in which $d_x=\gamma$.
\end{corollary}
\begin{proof}
Without loss of generality, assume $x\in B$. Then,
there are several cases, each determined by the
relative values of $\alpha$, $\beta$ and $\left\lfloor\mu\right\rfloor$.
The most general case is $\alpha<\left\lfloor\mu\right\rfloor<\beta$,
so consider only this situation.
Start from the realization with $d_x=\beta$.
Repeated applications of the method in the proof of Theorem~\ref{gratri} will eventually
yield a realization in which $d_x=\left\lfloor\mu\right\rfloor$.
For each step, the degree of
$x$ will have decreased by~1.
Therefore, one realization
of the triplet will have been found with $d_x=\gamma$ for
all $\left\lfloor\mu\right\rfloor\leqslant\gamma\leqslant\beta$.

Now, start from the realization with $d_x=\alpha$.
Applying the same step from the proof of Theorem~\ref{gratri} repeatedly will eventually
yield a realization in which $d_x=\left\lfloor\mu\right\rfloor$.
For each of these steps,
the degree of $x$ will have increased by~1. Therefore,
one realization of the triplet will have been found with $d_x=\alpha$ for
all $\alpha\leqslant\gamma\leqslant\left\lfloor\mu\right\rfloor$.
\end{proof}

Notice that, given a graphical triplet,
Corollary~\ref{cortri} also implies the existence of minimum
and maximum allowed degrees for each node whose
degree has not yet been  fixed in that triplet (namely, in $B$ and $K$). That is, a realization
of the triplet exists with a node having
either its minimum or maximum degree,
or any degree between these two values. Of course, the value of the minimum
and maximum degree will depend on which degrees have been fixed up to that
point, so these need to be computed on the fly. How to calculate these degree bounds
will be explained in Subsection \ref{ssec:desc}. 

\subsection{Building a degree-spectra matrix}\label{buidsm}
Corollary~\ref{cortri} suggests the possibility
of a direct, sequential way to build a degree-spectra
matrix from a JDM. However, building the degree-spectra matrix
node by node is a local process, which guarantees 
via Theorem \ref{gratri} only that
the bipartite graph in which the node whose degree spectrum is
being set resides is graphical.  There is a {\em global} constraint, 
however, on every node, namely that the sum of their degree 
spectra must add up to the degree of the class they belong to. 
We have to make sure that the local construction process also 
respects the global constraints, i.e., it is {\em feasible} with it.
The theorem below will show that this sequential construction process
is feasible, and just as importantly, all graphical realizations
of a JDM $J$ can be constructed in this way, i.e., all graphical degree-spectra
matrices can be obtained by this sequential construction process.

\begin{theorem}\label{decomp}
Let ${\cal S}$ be the subset of all the nodes with fixed
spectra; then, there exists a realization of a
JDM $J$ consistent with the fixed spectra if and only
if for every $(\alpha,\beta)$ pair with $\alpha,\beta \in \{1,\ldots,\Delta\}$
there exists a graph $G_{\alpha\beta}$ with $J_{\alpha,\beta}$
edges also satisfying the fixed spectra of ${\cal S}$.
\end{theorem}
\begin{proof}
Necessity is obvious. If there exists
a realization of $J$ satisfying the spectra,
then each subgraph between any pair of degree
classes both satisfies the spectra and has the right
number of edges.

To prove sufficiency, assume that we have a fixed degree
spectrum for all the nodes in ${\cal S}$ and we have guaranteed
the graphicality of
all the subgraphs $G_{\alpha\beta}$. 
They have the right
number of edges $J_{\alpha,\beta}$ and their
nodes satisfy the \textit{fixed} spectra specified
in the subset ${\cal S}$. Since we have guaranteed graphicality for all the
$G_{\alpha\beta}$ subgraphs with these constraints, let us consider some graphical 
realization for each such subgraph and consider their union graph $G$. 
If the ``free''
nodes, i.e., those without a fixed spectrum,
have all the correct degree in $G$ (i.e., every node $v \in V_{\alpha}$ has $d_v = \alpha$
for all $\alpha$), then there is nothing to do.
Now, assume they don't. Since the total number
of edges in each $G_{\alpha\beta}$ is correct by hypothesis,
there must exist a degree $\alpha$ and two free
nodes $v$ and $w$ belonging to $V_\alpha$ such that
$d_v<\alpha$ and $d_w>\alpha$.
Thus, there must exist a node
$u$ connected to $w$ but not to $v$. Then, erase
the edge $(u,w)$, and replace it with $(u,v)$.
This leaves the numbers of edges
in all $G_{\alpha\beta}$ unchanged, and does not change
the degree spectrum of $u$, because $v$ and $w$ belong
to the same degree class. Repeating this procedure
results eventually in all the nodes having the
correct degree.
\end{proof}

Theorem~\ref{decomp} is fundamentally important as it 
justifies a systematic, node-by-node
approach in building a graphical degree-spectra matrix.
In fact, so long as one guarantees the possibility
of subgraphs with the correct number of
edges, a partial degree-spectra matrix maintains
the graphicality of the JDM.

The only detail left is specifying how to choose
the numbers that form the degree spectra.
Fortunately, this is straightforward.
As mentioned in the previous Subsection,
an implication of Corollary~\ref{cortri}
is the existence of minimum and maximum allowed degrees
for nodes in partial degree sequences.
Let them be $m$ (minimum) and $M$ (maximum).
But a partial degree sequence is nothing
else than a partially built degree spectrum,
if one recognizes the node sets $U$ and $V$ as two degree
classes. Then, a condition that must
be satisfied in building a degree-spectra
matrix is that any new number chosen to augment a partially
built degree spectrum has to be within
these bounds. However, one must also consider that
if a node belongs to a certain degree class,
it must have the correct total degree.

To state both conditions, assume
the degree spectrum of node $v\in V_\alpha$
is being built.
Let $\Gamma$ be the set of degree classes for
which a spectrum element has already
been chosen, and let $S_{\beta v}$ be the
element to determine next.
Then, a valid value $k$ for $S_{\beta v}$
must satisfy the two conditions
\begin{eqnarray}
m_\beta\leqslant k\leqslant M_\beta\\
\sum_{\eta\notin\left(\Gamma\cup\beta\right)}m_\eta\leqslant \alpha-k-\sum_{\eta\in\Gamma}S_{\eta v}\leqslant\sum_{\eta\notin\left(\Gamma\cup\beta\right)}M_\eta\:.
\end{eqnarray}
Below, in Subsection \ref{ssec:desc} we describe how to compute the min and max values for degree
spectra elements.

\section{The algorithm}

\subsection{Description}\label{ssec:desc}

We are now ready to describe our JDM sampling algorithm.
The algorithm is composed of two parts. The first
is a spectra sampler that randomly generates degree-spectra matrices
from a graphical JDM $J$:
\begin{enumerate}
 \item Initialize $i=1$.
 \item Set $\alpha=1$.
 \item Let $l$ be the number of the residual, unallocated stubs of node $i$. If $l≠0$:
\begin{enumerate}
 \item If $J_{d_i,\alpha}≠0$:
\begin{enumerate}
 \item For all $\alpha\leqslant\beta\leqslant\Delta$, if $J_{d_i,\beta}≠0$, find $M_k$ and $m_k$; otherwise, set $m_k=M_k=0$.
 \item Compute $t=\sum_{\beta=\alpha+1}^\Delta m_\beta$ and $T=\sum_{\beta=\alpha+1}^\Delta M_\beta$.
 \item Find the actual minimum and maximum allowed for the degree-spectrum element:
$r=\max\left\lbrace m_\alpha, l-T\right\rbrace$ and $R=\min\left\lbrace M_\alpha, l-t\right\rbrace$.
 \item Extract an integer $S_{\alpha,i}$ uniformly at random between $r$ and $R$.
\end{enumerate}
 \item Increase $\alpha$ by~1, and go to step~(iii).
\end{enumerate}
 \item Increase $i$ by~1. If $i\leqslant N$, go to step~(ii).
\end{enumerate}
\begin{figure}
\centering
\includegraphics[width=0.75\textwidth]{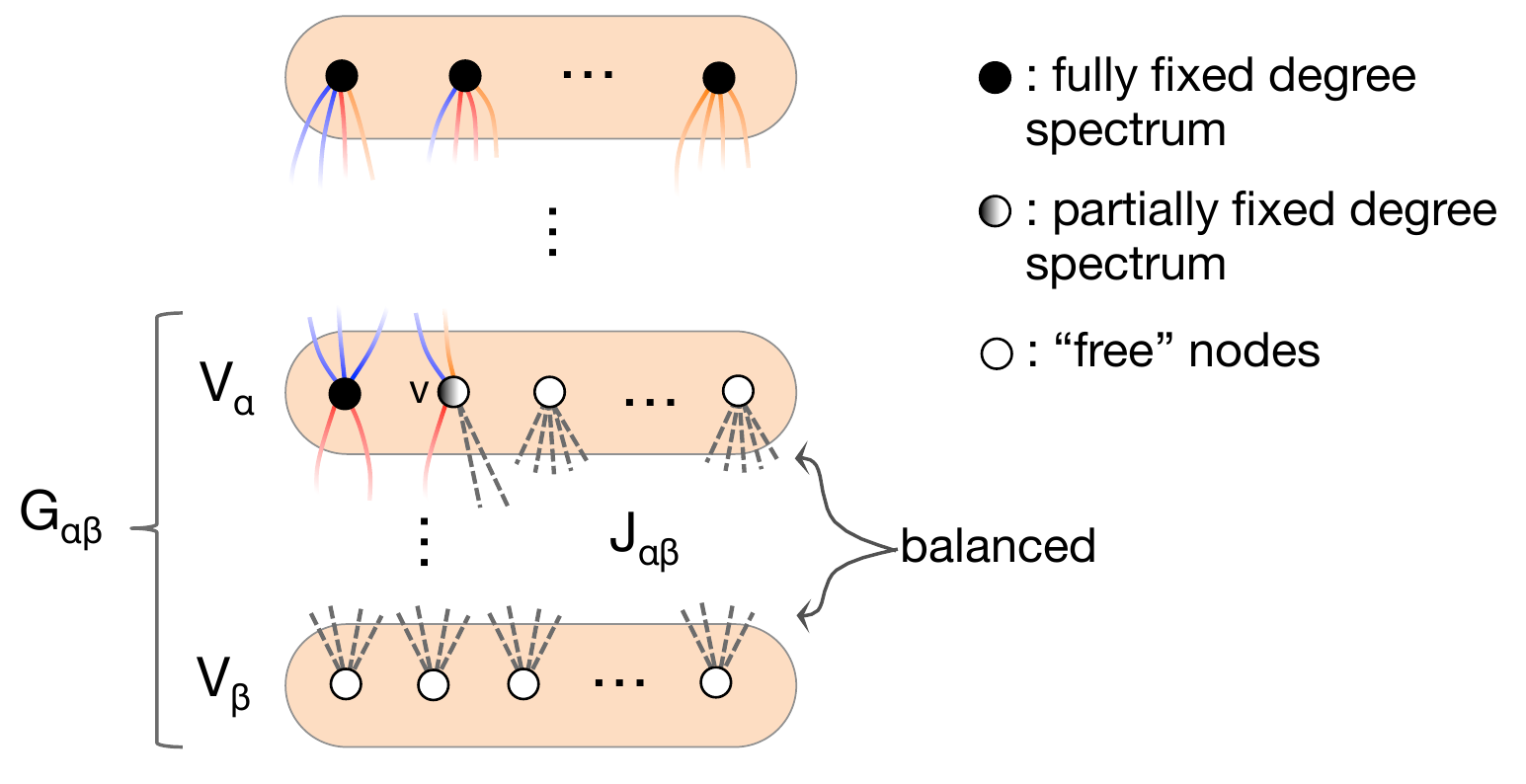}
\caption{\label{fig:spectrum} Sequentially determining graphical degree spectra 
consistent with a given JDM $J$. }
\end{figure}

To find the values of $m$ and $M$ in step (iii).a.1 above, consider
the degrees of the nodes belonging to $V_{\alpha}$
and $V_{\beta}$ in $G_{\alpha\beta}$. In the formalism of
Subsection~\ref{degspe}, the already
fixed spectra elements are equivalent to the sequences $\mathcal{P}$
and $\mathcal{Q}$. Then, to test the viability of a given value
as a degree-spectrum element, assign it to the element being determined,
complete the degree sequence making it balanced, and test it for
graphicality, see Fig. \ref{fig:spectrum}. If the sequence is graphical, then the triplet has
a balanced realization, which by Theorem~\ref{gratri} is a necessary
and sufficient condition for the existence of a subgraph corresponding
to the spectrum element being determined. If $G_{\alpha\beta}$ is unipartite,
the graphicality test can be done using the
fast method described in~\cite{Del10}.
The situation is marginally different
if $G_{\alpha\beta}$ is bipartite. In this case,
as previously mentioned, the degree
sequence can be built as a BDS in which
nodes of degree $\alpha$ only have
incoming edges, and nodes of degree
$\beta$ only have outgoing ones.
This sequence can then be tested with the
fast directed graphicality test described in~\cite{Kim12}.

Thus, to find the minimum value $m$ one can simply run a sequential
test, checking for valid spectrum values from~0 onwards. The
first successful value is $m$. Then, to
find $M$, use bisection to test all the values
from $m+1$ to the theoretical maximum, looking for the
largest number allowed. Clearly, the theoretical maximum at that stage is the
degree of the class the node belongs to 
minus the sum of the already fixed spectra values for that degree.

These considerations also clarify the nature
of the second part of the algorithm, which
samples realizations of the JDM from an extracted degree
spectra matrix. Summarizing,
\begin{itemize}
 \item JDM realizations can be decomposed into a set of independent unipartite and bipartite graphs.
 \item The degree spectra define the degree sequences of the component subgraphs.
\end{itemize}
Then, to accomplish the actual sampling,
extract the degree sequences from the degree
spectra and use them in the graph sampling algorithms
for undirected and directed graphs presented in~\cite{Del10,Kim12} and in here in \ref{apxA}.
Every time a sample is generated, it constitutes a subgraph
of a JDM realization. All that is needed in the end
is simply to list the edges correctly, since
the graph realizing the JDM is the union of all
the unipartite and bipartite subgraphs into
which it has been decomposed.

\subsection{Sampling weights}

Our algorithm does not extract
all degree-spectra matrices from a JDM with
the same probability. However, the
relative probability for the extraction of
each spectra matrix is easily computed,
and it can be used to reweight the sample
and obtain unbiased sampling.
If every new
element of a degree-spectra matrix is
extracted uniformly at random between $r$ and $R$,
its probability of being chosen is simply $\frac{1}{R-r+1}$.
Therefore, the probability of extracting a
given spectra matrix $S$ is $p\left(S\right)=\prod_{i=1}^m\frac{1}{R-r+1}$,
where $m$ is the total number of elements extracted.
Then, an unbiased estimator for a network
observable $Q$ on an ensemble of $Z$ spectra
matrices can be computed using the weighted average
\begin{equation}\label{degspest}
 \left\langle Q\right\rangle = \frac{\sum_{i=1}^Z Q_i w_i}{\sum_{i=1}^Z w_i}\:.
\end{equation}
In the expression above, $Q_i$ is the value that $Q$
assumes on the $i^\mathrm{th}$ sampled matrix. Indicating by
$r_j$ and $R_j$ the values that $r$ and $R$ assume
for the $j^\mathrm{th}$ matrix element extracted, the
weights are
\begin{equation}\label{weights}
w_i=\prod_{j=1}^m\left(R_j-r_j+1\right)\:.
\end{equation}
\begin{figure}
\centering
{\includegraphics[width=0.6\textwidth]{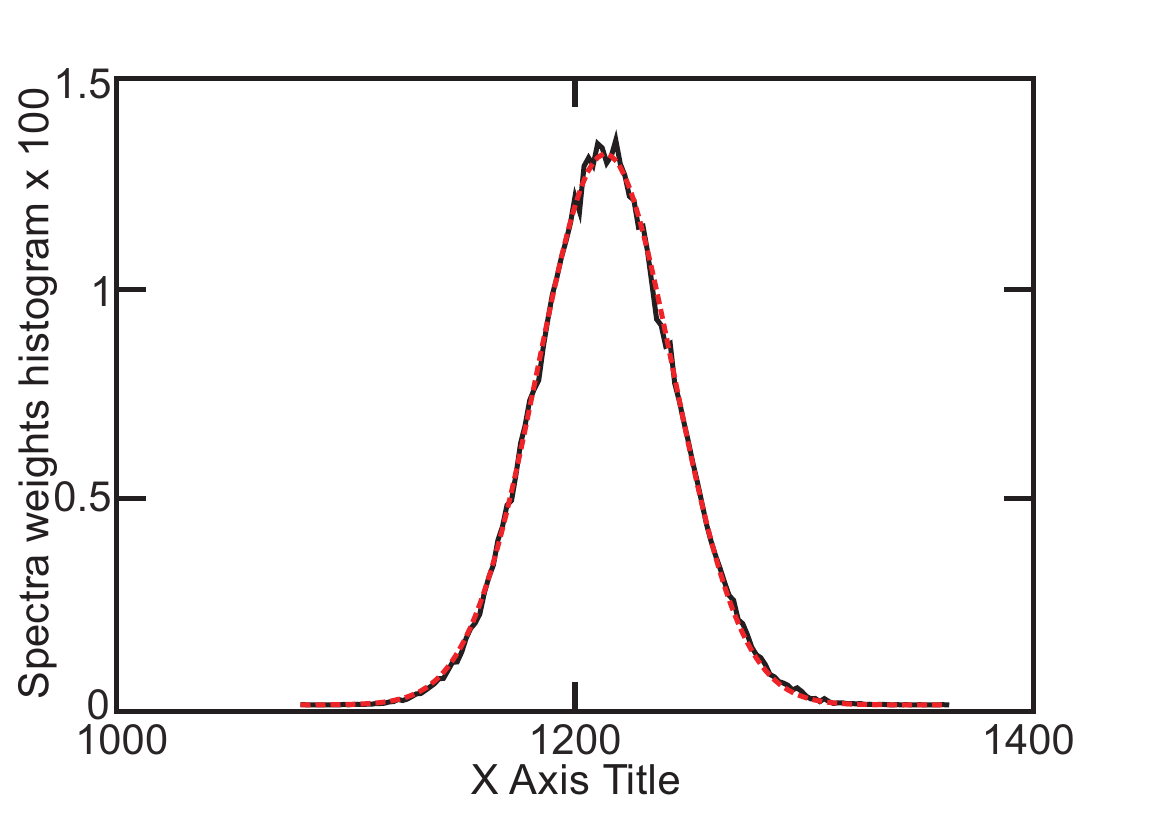}\vspace*{-18pt}}
\includegraphics[width=0.6\textwidth]{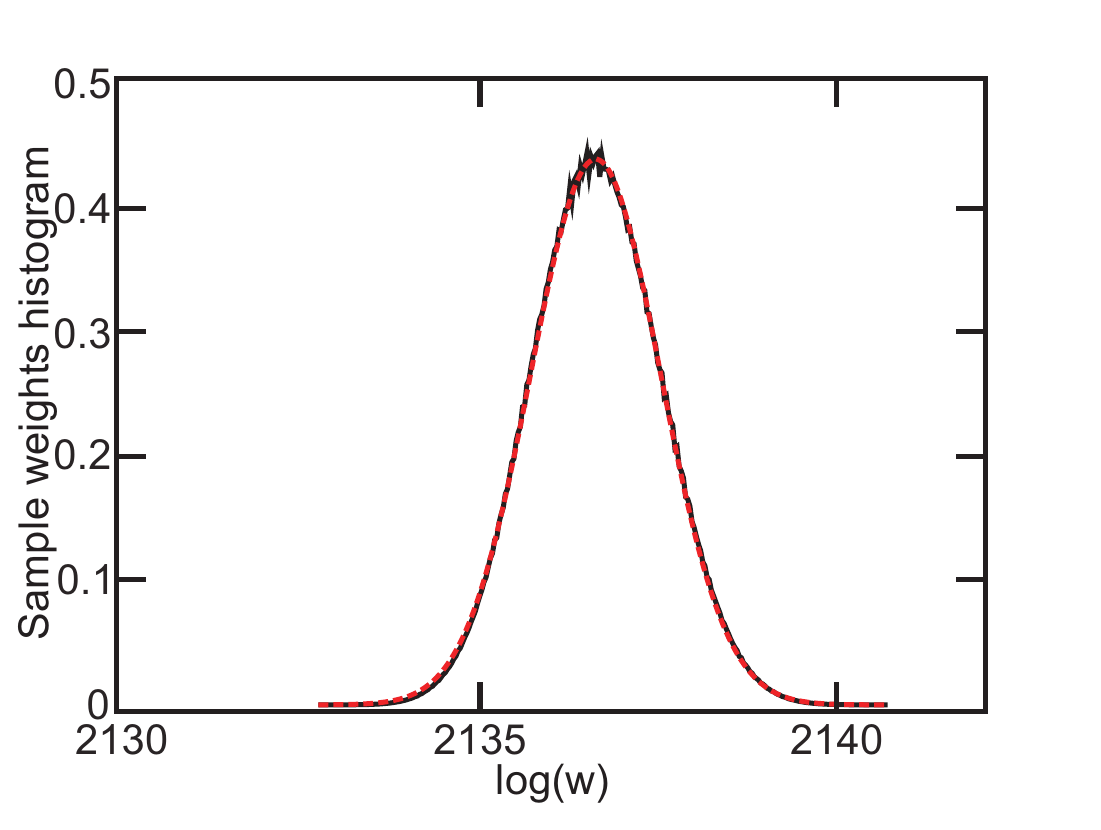}
\caption{\label{Fig1}Log-normal distribution of weights.
The top panel shows the histogram of the natural logarithms
of the weights for an ensemble of $10^5$ degree-spectra
matrices; the bottom panel shows the histogram for an ensemble
of $10^8$ sample weights. Both distributions (solid black
lines) are well fitted by a Gaussian curve (dashed red line).}
\end{figure}

Of course, besides the spectra matrix, every subgraph
has its own sampling weight. Thus, the total weight
of a single JDM sample is the product of the corresponding
spectrum weight and all the subgraph weights. To describe
the distribution of the sample weights, first recall
that the individual subgraph weights are log-normally
distributed~\cite{Del10,Kim12}. Thus, as the sample
weights are their product, we expect them to be log-normally
distributed too. Also, for large JDMs, where $\Delta^2\gg1$,
the $m$ factors in Eq.~\ref{weights} are effectively
random. Thus, our expectation is that the spectra weights
are log-normally distributed as well. To verify this,
we extracted the JDM of a random scale-free network
with 1000~nodes and power-law exponent of~$2.5$,
and used it to generate an ensemble of $10^5$ degree
spectra matrices and one of $10^8$ JDM samples of
a single spectra matrix. Figure~\ref{Fig1} shows that
the histograms of the logarithms of spectra matrix
weights and sample weights are well approximated
by a Gaussian fit, supporting our assumptions.

A simple and small example is provided in~\ref{apxB}.
There, we analytically compute the JDM ensemble averages
of the local clustering coefficients of nodes of all
degrees, based on unweighted sampling and also based
on weighted sampling, with the weights provided by the
algorithm. In table~\ref{valid}, we show the results
of simulations using our algorithm, taking into account
the sample weights (as described above), and simply computing
the averages of the clustering coefficients over the
samples generated. The results between theoretical and
simulated measures agree very well. The differences between
weighted and unweighted versions can be also appreciated,
and while they are small in this example, they are measurable
and need to be taken into account in general.

\subsection{Computational complexity}

To determine the computational complexity
of the algorithm, first note that the
main cost in creating a spectra matrix comes
from the repeated graphicality tests. Let
$A$ be the number of non-empty degree classes
in the JDM
\begin{equation*}
 A=\left|\left\lbrace \alpha:\ V_\alpha≠\emptyset\right\rbrace\right|\:.
\end{equation*}
Then, for each of the $NA$ non-trivial elements
in the degree-spectra matrix, $A$ tests are needed,
each with a computational complexity
of the order of the number of nodes in the corresponding
degree class. Thus, the total computational complexity
for the spectra construction part of the algorithm
is
\begin{equation}\label{compsp}
C_S = {\cal O}\left(N\sum_{\alpha=1}^\Delta\sum_{\beta=\alpha}^\Delta\left|V_\beta\right|\right)\:.
\end{equation}
Notice that in our treatment one is free to choose the order
of the degree classes. Thus, to minimize the complexity, one
can simply determine the degree-spectra elements in descending
order of degree class size. Then, the worst case corresponds
to the equipartition of the nodes amongst degree classes, 
$\left|V_\alpha\right|=\frac{N}{A}$.
In this case, it is
\begin{equation*}
C_S={\cal O}\left(NA^2\frac{N}{A}\right)={\cal O}\left(N^2A\right)\:,
\end{equation*}
which reduces to
\begin{equation*}
C_S={\cal O}\left(N^3\right)
\end{equation*}
if the number of degree classes is of the same order
as the number of nodes. 

A more precise estimate for a given JDM
can be obtained by rewriting Eq.~\ref{compsp} as
\begin{equation*}
 C_S = {\cal O}\left(N^2\sum_{\alpha=1}^\Delta\sum_{\beta=\alpha}^\Delta P\left(\beta\right)\right)\:,
\end{equation*}
where the \emph{degree distribution} $P\left(d\right)=\left|V_d\right|/N$
is the probability that a randomly chosen node has degree $d$. It is easy to
see, then, that the worst case is unlikely to occur. Consider for instance systems
of widespread insterest, such as scale-free networks, 
for which $P\left(d\right)\sim d^{-\gamma}$ with $\gamma > 2$.
Then, in the limit of large networks, the equation above becomes
\begin{equation*}
 C_S= {\cal O}\left(N^2\int_1^\infty\mathrm{dx}\int_x^\infty\mathrm{dk}\left(\gamma-1\right)k^{-\gamma} \right)
= {\cal O}\left(\frac{N^2}{\gamma-2}\right)
= {\cal O}\left(N^2\right)\:.
\end{equation*}
Thus, in this case,
the complexity leading order for spectra
matrix extraction is only quadratic.

Given a degree-spectra matrix, to construct a JDM realization
one then needs to build ${\cal O}\left(A^2\right)$ subgraphs,
each with ${\cal O}\left(\frac{N}{A}\right)$ nodes and ${\cal O}\left(\frac{M}{A}\right)$
edges. For each subgraph, the computational complexity is of
the order of the number of nodes multiplied by the number of
edges. Thus, the total sampling complexity is
${\cal O}\left(A^2\frac{N}{A}\frac{M}{A}\right)={\cal O}\left(NM\right)$.
Therefore, the total complexity of the graph construction part
of our method is ${\cal O}\left(N^2\right)$ for sparse
networks, and ${\cal O}\left(N^3\right)$ for dense ones.
Once more, we do not expect the worst case complexity to occur
often. For example, in the already mentioned case of scale-free
networks, which are always sparse~\cite{Del11}, the total complexity
of our algorithm would only be quadratic. A less efficient sampling
method has been developed recently~\cite{Gjo15}, but it is based
on backtracking, producing results containing biases that are
uncontrolled and that cannot be estimated.

\section{Conclusions}

In summary, we have solved the problem of constrained
graphicality when degree correlations are specified,
developing an exact algorithm to construct
and sample graphs with a specified joint-degree matrix.
A JDM specifies the number of edges that occur between
degree classes of nodes (nodes of given degrees), and 
thus completely determines
all pairwise degree correlations in its realizations.
Our algorithm is guaranteed to successfully build a random
JDM sample in polynomial time, systematically, and without
backtracking. It is also guaranteed to be able to build any of the
graphical realizations of a JDM. Each graph is constructed independently and thus
there are no correlations between samples.
Although the algorithm does introduce a sample bias,
the relative probability for the construction of each
sample is computable, which allows the use of weighted
averages to obtain unbiased sampling (importance sampling).
However, importance sampling is only exact in the limit of an infinite number of
samples. This raises the issue of convergence. The lognormal distribution of
weights makes convergence slow, but for small- to medium-sized networks good
accuracy can be achieved, and quantities computed as if from uniform sampling.
Improving the speed of convergence is a challenging problem,
partly because it depends on the constraining JDM, and will be addressed in future publications.

Degree correlations in real-world systems have been widely observed.
Social networks are known to be positively correlated, and the concept
of assortativity was known to the sociological literature before it
was employed in applied mathematics. Technological networks are also
characterized by particular correlation profiles. Moreover, correlations
significantly affect the dynamics of spatial processes, such as the
spread of epidemics~\cite{Boc06}. Thus, with our algorithm, one can model
correctly complex systems of general interest with desired degree assortativity.
For the first time, this enables the study of networks in which the
correlations are not determined solely by the nodes' degrees.
For instance, there exist many studies about social networks, consisting
of a comparison between some specific real-world network and a randomized
ensemble of networks with the same degree sequence or degree distribution.
As social networks are scale-free, these studies often just sample
the same sequence or the same type of power-law sequences to produce
null-model results. However, social networks are assortative, while random
scale-free networks are on average disassortative. Thus, the average
correlations of scale-free networks make degree-sequence and degree-distribution
sampling problematic if one is trying to consider a random model of
a social network. Our method allows one to avoid this problem by directly
imposing the correlations, rather obtaining only those imposed by the degree
sequence.

Upper bounds on the computational complexity of our algorithm
show that in the worst case it is cubic in the number of nodes.
However, we provide a way to compute the expected worst-case complexity
if the degree distribution of the networks considered is known.
This shows that, for commonly studied cases such as scale-free
networks, the maximum complexity is only of the order of $N^2$,
making the algorithm even more efficient.

\ack
KEB, PLE, IM, and ZT acknowledge support by
the AFOSR and DARPA through Grant FA9550-12-1-0405.
KEB also acknowledges support from the NSF through Grant
DMR-1206839. CIDG acknowledges support by EINS, Network
of Excellence in Internet Science, via the European
Commission's FP7 under Communications Networks, Content
and Technologies, Grant 288021 and ZT also acknowledges 
support from DTRA through Grant HDTRA-1-09-1-0039.

\appendix
\section{Direct construction of random 
directed and undirected graphs with prescribed degree sequence} \label{apxA}

In order to fully describe our algorithm for sampling graphs
with prescribed degree correlations, we include in this appendix
succinct descriptions of our algorithms for sampling
random undirected~\cite{Del10} and directed~\cite{Kim12} graphs
with a prescribed degree sequence. Both are used in our algorithm
to sample graphs with a prescribed JDM, and both work by directly
constructing the graphs. So long as the prescribed degree sequence
is graphical, both algorithms are guaranteed to successfully
construct a graph without backtracking. They accomplish this
by building the graph an edge at a time, connecting pairs of
stubs, maintaining the graphicality of the residual stubs throughout
the construction process. The algorithms make use of our 
fast methods for testing the graphicality of degree sequences, which
are also described below. The worst case complexity is
${\mathcal O}\left(N\right)$ for the graphicality tests,
and ${\mathcal O}\left(N M \right)$ for both sampling algorithms.
Both algorithms generate biased samples, but we also state the relative
probability of generating a sample, which can be used to calculate
unbiased statistical averages. See our previous publications for
proof of the correctness of these algorithms~\cite{Del10,Kim12};
they are stated without proof or detailed explanation here.

\subsection{Undirected graphs}

A nonincreasing sequence of integers 
${\mathcal D}=\left\lbrace d_1, d_1, \dots, d_{N}\right\rbrace$
is graphical if and only if 
$\sum_{i=1}^{N}d_i$
is even, and $L_k\leqslant R_k$ for all $1\leqslant k< N$,
where $L_k$ and $R_k$ are given by the recurrence relations
\begin{eqnarray}
 L_1 &= d_1\\
 L_k &= L_{k-1}+d_k
\end{eqnarray}
and
\begin{eqnarray}
 R_1 &= N-1\\
 R_k &=\left\lbrace\begin{array}{l}
                    R_{k-1}+x_k-2\quad\ \,\forall k<k^\ast\\
		    R_{k-1}+2(k-1)-d_k\quad\forall k\geqslant k^\ast
                   \end{array}\right.
\end{eqnarray}
and we defined the \emph{crossing indices}
$x_k=\min\left\lbrace i: d_i < k\right\rbrace$,
and $k^\ast = \min\left\lbrace i: x_i < i+1\right\rbrace$. 
Thus, to test the graphicality of ${\mathcal D}$:
\begin{enumerate}
\item Sum the degrees to determine if
$\sum_{i=1}^{N}d_i$ is even. If false, then stop; ${\mathcal D}$ is not graphical. If true, continue.
While summing the degrees, also calculate the crossing indices $x_k$ for each $k$ and determine $k^*$.
\item Test if $L_1 \leq R_1=N-1$. If false, then stop; ${\mathcal D}$ is not graphical. If true, set $k=2$ and continue.
\item Test if $L_k \leq R_k$. If false, then stop; ${\mathcal D}$ is not graphical. If true, increase $k$ by one and repeat. 
Continue until $k=N-1$, then stop; ${\mathcal D}$ is graphical.
\end{enumerate}

Given a nonincreasing graphical degree sequence ${\mathcal D}$,
a random undirected graph that realizes ${\mathcal D}$ can be constructed by:
\begin{enumerate}
\item To each node, assign a number of stubs equal to its degree.
 \item Choose a hub node $i$. Any node can in principle be chosen, for example, the node with the largest degree.
 \item Create a set of forbidden nodes $X$, which initially contains only $i$.
 \item Find the set of allowed nodes $A$ to which $i$ can be linked
preserving the graphicality of the remaining construction process.
To find $A$, first determine the \emph{maximum fail degree} $\kappa$
using the method described below. Then $A$ will consist of all nodes
$j \notin X$ that have remaining degree greater than $\kappa$.
 \item Choose a random node $m \in A$ and connect $i$ to it. 
 \item Reduce the value of $d_i$ and $d_m$ in ${\mathcal D}$ by 1, and reorder it.
 \item If $m$ still has unconnected stubs, add it to the set of forbidden nodes $X$.
 \item If $i$ still has unconnected stubs, return to step~(iv).
 \item If nodes still have unconnected stubs, return to step~(ii).
\end{enumerate}

To determine the maximum fail degree in a degree sequence $\mathcal{D}$ being sampled,
build the \emph{residual degree sequence}
$\mathcal{D}'$, 
by connecting
the hub node $i$ with remaining degree $d_i$ to the $d_i-1$ nodes with the largest degrees that
are not in the forbidden set $X$ and reducing the elements of ${\mathcal D}$ accordingly.
Then, compute the graphicality test inequalitites.
Each inequality potentially yields a fail-degree candidate,
depending on the values of $L_k$ and $R_k$. For each
value of $k$ there are only 3 possibilities:
\begin{enumerate}
 \item[(a)] $L_k=R_k$
 \item[(b)] $L_k=R_k-1$
 \item[(c)] $L_k\leqslant R_k-2$
\end{enumerate}
In case~(a), the degree of the first non-forbidden node whose index is greater than $k$ is the
fail-degree candidate.
In case~(b), the degree of the first non-forbidden node whose index is greater
than $k$ and whose degree is less than $k+1$ is the fail-degree candidate.
In case~(c), there is no fail-degree candidate.
The sequence of candidate nodes
is non-decreasing until the fail-degree is found. Thus,
one can stop the calculation when either the current
fail-degree candidate is less than the previous one,
or when a case~(a) happens.

This algorithm generates graph samples biasedly. 
However, the relative
probability of generating a particular sample $\mu$ is
\begin{equation}
 p_\mu = \prod_{i=1}^m \bar d_i!\prod_{j=1}^{\bar d_i}\frac{1}{\left|A_{i_j}\right|}\:,
\end{equation}
where 
$\bar d_i$ is the residual degree
of node $i$ when it is chosen as a hub,
$m$ is the total number of hubs used,
and $A_{i_j}$ is the allowed set for the $j^\mathrm{th}$
link of hub $i$.
Thus, an unbiased estimator
for a network observable $Q$ for any target distribution $P$ is the weighted
average
\begin{equation}
 \left\langle Q\right\rangle = \frac{\sum_{i=1}^M Q_{\mu_i} w_{\mu_i} P\left(\mu_i\right)}{\sum_{i=1}^M w_{\mu_i} P\left(\mu_i\right)}\:,
 \label{unbiasedestimator}
\end{equation}
where $M$ is the number of samples and $w_{\mu_i}=p_{\mu_i}^{-1}$.
For uniformly sampling
the networks, $P$ is constant and it cancels
out of the formula.

\subsection{Directed graphs}

A bi-degree sequence (BDS)  $\mathcal{D} = \left\lbrace\left(d_1^-, d_1^+\right), \left(d_2^-, d_2^+\right), \dots, \left(d_N^-, d_N^+\right)\right\rbrace$
of integer pairs, ordered so that the first
elements of each pair form a non-increasing sequence, is graphical if and only if
$\sum_{i=1}^N d_i^- = \sum_{i=1}^N d_i^+$,
and $L_k\leqslant R_k$ for all $1\leqslant k\leqslant N-1$,
where $L_k$ and $R_k$ are given by the recurrence relations
\begin{eqnarray}
 L_1 &= d_1^-\\
 L_k &= L_{k-1}+d_k
\end{eqnarray}
and
\begin{eqnarray}
 R_1 &= N-1-G_1\left(0\right)\\
 R_k &= \left\lbrace\begin{array}{l}
                    R_{k-1} + N - \bar G_{k-1}\left(k-1\right)\quad\quad\forall d_k^+<k\\
		    R_{k-1} + N - \bar G_{k-1}\left(k-1\right)-1\quad\forall d_k^+\geqslant k\\
                   \end{array}\right.\:,
\end{eqnarray}
and $G_k$ and $\bar G_k$ are defined as follows.
Let
\begin{equation}
 g_i\left(k\right) = \left\lbrace \begin{array}{l}
                      d_i^+ + 1\quad\forall i\leqslant k\\
		      d_i^+\quad\quad \forall i>k\\
                     \end{array} \right. \:.
\end{equation}
Then
\begin{equation}
 G_k\left(p\right) = \sum_{i=1}^N \delta_{p,g_i\left(k\right)}\:,
\end{equation}
where $\delta$ is the Kronecker delta,
and
$\bar G$ is given by the recurrence relation
\begin{eqnarray}
 \bar G_1\left(1\right) &= G_1\left(0\right) + G_1\left(1\right)\\
 \bar G_k\left(k\right) &= \bar G_{k-1}\left(k-1\right) + G_1\left(k\right) + S\left(k\right)\:,
\end{eqnarray}
where
\begin{equation}
 S\left(k\right) \equiv \sum_{t=2}^{k-1}\delta_{k,d_t^++1} - \sum_{t=2}^{k}\delta_{k,d_t^+}\:.
\end{equation}
To efficiently test the graphicality of a BDS ${\mathcal D}$, 
\begin{enumerate}
\item Sum the in- and out-degrees to determine if
$\sum_{i=1}^N d_i^- = \sum_{i=1}^N d_i^+$. If false,
then stop; ${\mathcal D}$ is not graphical. If true, continue.
While summing the degrees, also calculate $L_k$ for each $k$.
\item Compute $G_1\left(k\right)$ for each $k$.
\item Compute $S\left(k\right)$ for all $k$:
\begin{enumerate}
 \item Initialize $S\left(k\right)$ to 0 for all $k$. Set $i=2$.
 \item If $d_i^+\geqslant i$, decrease $S\left(d_i^+\right)$ by~1.
 \item If $d_i^++1> i$, increase $S\left(d_i^++1\right)$ by~1.
 \item Increase $i$ by~1. If $i\leqslant N$, repeat from step~(b).
\end{enumerate}
\item Test if $L_1 \leq R_1$. If false, then stop; ${\mathcal D}$ is not graphical. If true, set $k=2$ and continue.
\item Test if $L_k \leq R_k$. If false, then stop; ${\mathcal D}$ is not graphical. If true, increase $k$ by one and repeat. 
Continue until $k=N-1$, then stop; ${\mathcal D}$ is graphical.
\end{enumerate}

Given a graphical BDS of integer pairs 
${\mathcal D}$ in lexicographic order,
a random directed graph that realizes ${\mathcal D}$ can be constructed by
\begin{enumerate}
\item Assign in-stubs and out-stubs to each node according to its degrees.
 \item Define as current hub the lowest-index node $i$ with non-zero out-degree.
 \item Create a set of forbidden nodes $X$, which initially contains $i$
and all nodes with zero in-degree.
 \item Find the set of allowed nodes $A$ to which an out-stub of $i$ can be connected
without breaking graphicality. To find $A$, first determine the \emph{maximum fail in-degree}
$\kappa$ using the method described below. Then $A$ will consist of all nodes
$j \notin X$ that have remaining in-degree greater than $\kappa$.
  \item Choose a random node $m \in A$ and connect an out-stub of $i$ to one of its in-stubs.
 \item Reduce the value of $d_i^+$ and $d_m^-$ in ${\mathcal D}$ by 1, and reorder it accordingly.
 \item Add $m$ to the set of forbidden nodes $X$.
 \item If $i$ still has unconnected out-stubs remaining, return to step~(iv).
 \item If nodes still have unconnected out-stubs, return to step~(ii).
\end{enumerate}

The following simple procedure can be used to efficiently find the fail-in-degree in step~(iv) of the sampling algorithm.
\begin{enumerate}
 \item Create a new BDS $\mathcal{D}'$
obtained from $\mathcal{D}$ by reducing the in-degrees of
the first $d_i^+-1$ non-forbidden nodes by~1,
and reducing the out-degree of $i$ to~1.
 \item If $i=1$, set $k=2$; otherwise, set $k=1$.
 \item Compute ${L_k}$ and ${R_k}$ of the BDS $\mathcal{D}'$. 
\item If ${L_k}\neq{R_k}$: increase $k$ by 1; if $k=N$, there is no fail-in-degree,
and all the non-forbidden nodes are allowed, so stop; otherwise, go to step~(iii).
 \item Find the first non-forbidden node in $\mathcal{D}'$ whose index is greater than $k$.
 \item Identify this node in the original BDS $\mathcal{D}$. Its in-degree is the fail-in-degree. Stop.
\end{enumerate}

As in the case of the sampling algorithm for undirected graphs,
this algorithm generates directed graph samples biasedly. However,
an unbiased estimator for a network observable $Q$ for any target
distribution $P$ is the weighted average given by Eq.~\ref{unbiasedestimator}.
In this case the weights are
\begin{equation}
 w_\mu = \prod_{i=1}^\nu \prod_{j=1}^{d_i^+}\left|A_{i_j}\right|\:,
\end{equation}
where $\nu$ is the total number of hubs used, $\left|A_{i_j}\right|$
is the size of the allowed set immediately before placing the $j^\mathrm{th}$
connection coming from the $i^\mathrm{th}$ hub, and $d_i^+$ is the
out-degree of the $i^\mathrm{th}$ node chosen as a hub.
Note that, unlike the case for undirected networks,
there is no factorial combinatorial factor in the weights. This is because
while the particular sequence of hub nodes chosen depends on
the links placed, every node with non-zero out-degree
will be selected, sooner or later, as the hub.
Therefore, all the samples produced would have
an extra, identical, multiplicative factor of $\prod_{i=1}^N\frac{1}{d_i^+!}$.
As only the relative probabilities are needed for estimating an observable,
and this factor is the same for every possible sample, it
is eliminated from the formula for the weights.

\section{An explicit example} \label{apxB}

To illustrate the sampling mechanism and the difference between
weighted and unweighted estimation, we consider the realizations of
the JDM
\begin{equation}\label{Jsmall}
 J = \left(\begin{array}{ccc}
      0 & 0 & 0\\
      0 & 2 & 4\\
      0 & 4 & 1\\
     \end{array}\right)\:,
\end{equation}
and explicitly compute the average local clustering coefficients
$\la c_d\ra$ of the nodes of degree $d$, for
all values of $d$. This JDM induces the degree sequence
${\mathcal D}=\left\lbrace 2, 2, 2, 2, 3, 3\right\rbrace$, and,
up to isomorphism, has only three possible realizations, shown in
Fig.~\ref{Jreal}.
\begin{figure}
\centering
\includegraphics[width=0.33\textwidth]{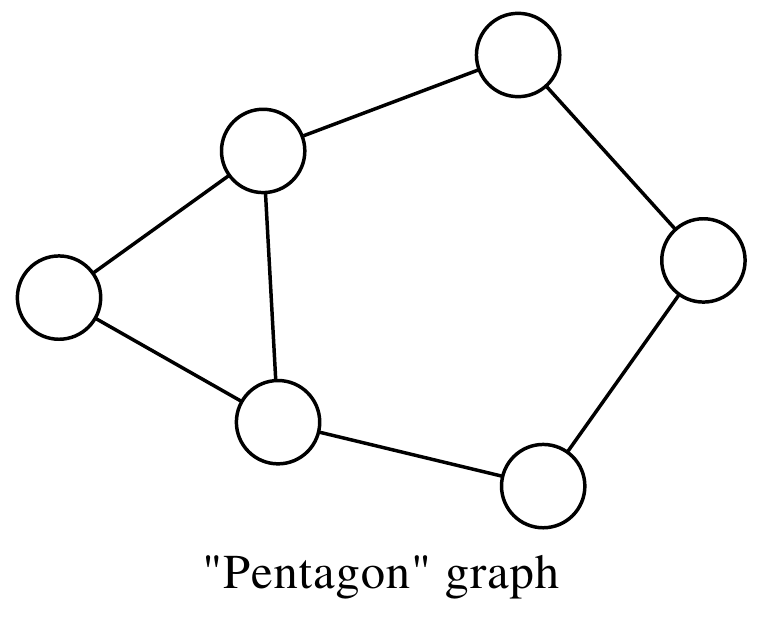}
\includegraphics[width=0.33\textwidth]{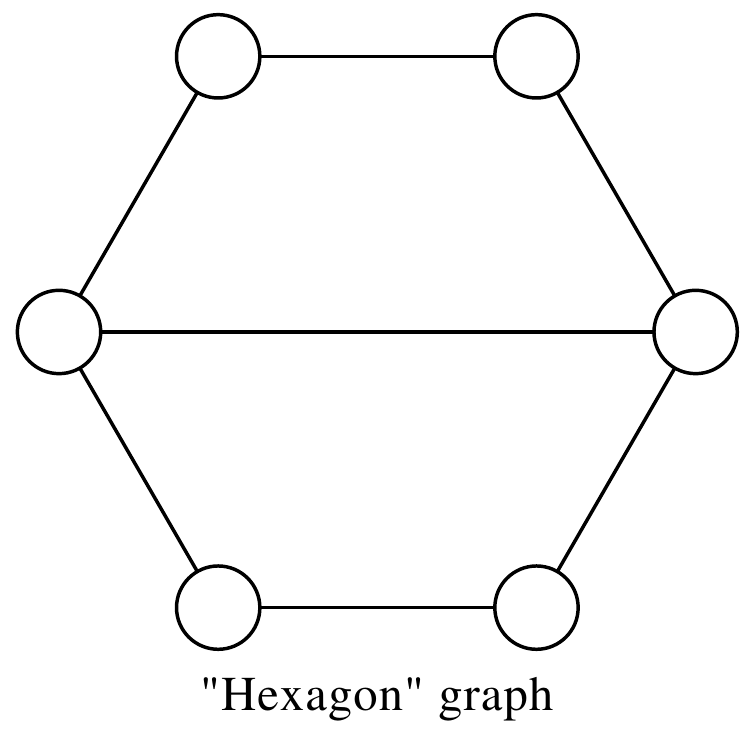}
\includegraphics[width=0.37\textwidth]{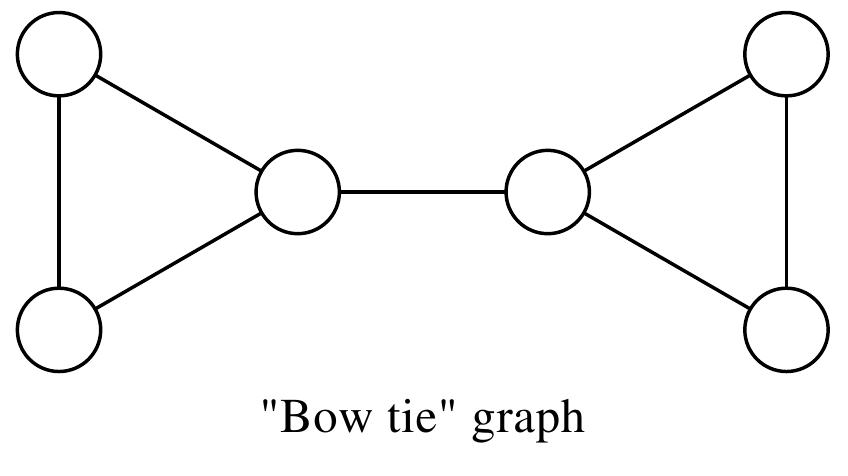}
\caption{\label{Jreal}Possible realizations of the JDM in Eq.~\ref{Jsmall},
up to isomorphism.}
\end{figure}
From the figures, it is easy to see that, for the pentagon
graphs, $\la c_2\ra_P = 1/4$ and $\la c_3\ra_P = 1/3$. Also, for
the hexagon graphs $\la c_2\ra_H = \la c_3\ra_H = 0$, while for
the bow tie graphs $\la c_2\ra_B = 1$ and $\la c_3\ra_B = 1/3$.

\subsection{Unweighted estimate}

To calculate the theoretical results for the unweighted case,
we need to consider the probability with which our algorithm
generates each degree-spectra matrix from~$J$. To this
purpose, first note that there are several degree-spectra
matrices whose realizations are all pentagon graphs. Also,
all the hexagon and bow tie graphs have the same degree-spectra
matrix
\begin{equation}\label{shb}
 S_{HB} = \left(\begin{array}{cccccc}
      0 & 0 & 0 & 0 & 0 & 0\\
      1 & 1 & 1 & 1 & 2 & 2\\
      1 & 1 & 1 & 1 & 1 & 1\\
     \end{array}\right)\:.
\end{equation}
This allows us to compute just the probability
of generating $S$, as all the other matrices will
yield the same contribution to $\la c_2\ra_P$ and
$\la c_3\ra_P$.

Our method chooses the elements of the degree-spectra
matrix $S$ being created in a systematic way, node by
node. As there are no nodes of degree~1, all the element
in the first row of the matrix are fixed to~0. Then,
the first element to choose is $S_{2,1}$, that is, the
number of edges between node~1 and nodes of degree~2.
The possible choices for this element are~0, 1, and~2.
Choosing~0 or~2 will result necessarily in a degree-spectra
matrix whose realizations are all pentagon graphs. In
fact, from Fig.~\ref{Jreal} one can see that, amongst
the realizations of~$J$, the pentagon graphs are the
only ones in which a node of degree~2, such as node~1,
has either no edges or 2 edges with nodes of degree~2.
Thus, choosing the value of $S_{2,1}$ with uniform probability,
at this stage one generates pentagon graphs with probability
$2/3$.

The remaining choice, $S_{2,1}=1$, happens with probability
$1/3$. In this case, $S_{3,1}$ is forced to be~1, since the
elements in the first column of $S$ must sum up to the degree
of the node~1, which is~2. The next element to determine is
then $S_{2_2}$. Similarly to the previous case, the possible
values are~0, 1, and 2. Choosing~0 or~2 will always result
in pentagon graphs, whose probability of being generated increases
by $1/3\cdot2/3=2/9$.

Choosing $S_{2,2}=1$, which occurs with total probability
$1/3\cdot 1/3=1/9$, forces $S_{3,2}=1$. The next value to
determine is that of $S_{2,3}$. As before choosing~0 or~2
yields pentagon graphs, whose total probability of being
generated increases by $1/3\cdot 1/3\cdot 2/3=2/27$.

The choice of $S_{2,3}=1$, which has a total probability
$1/3\cdot 1/3\cdot 1/3=1/27$ of happening, implies that
$S_{2,3}=1$. Then, the degree-spectra matrix being built
can only be $S_{HB}$. In fact, as it is evident from Fig.~\ref{Jreal},
the only graphs realizing~$J$ in which at least~3 nodes
of degree~2 are linked exactly to one other node of degree~2
and one of degree~3, are hexagon and bow tie graphs.

This shows that the degree-spectra matrix $S_{HB}$
occurs with probability $1/27$; conversely, degree-spectra
matrices yielding pentagon graphs occur with probability
$26/27$.

\begin{figure}
\centering
\includegraphics[width=0.15\textwidth]{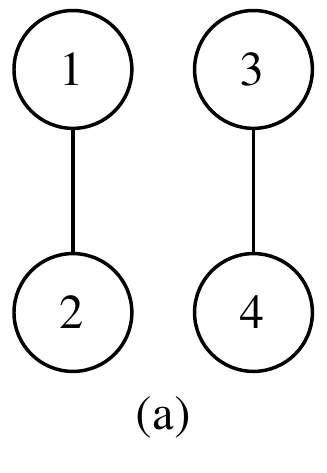}\quad\quad\quad\quad
\includegraphics[width=0.15\textwidth]{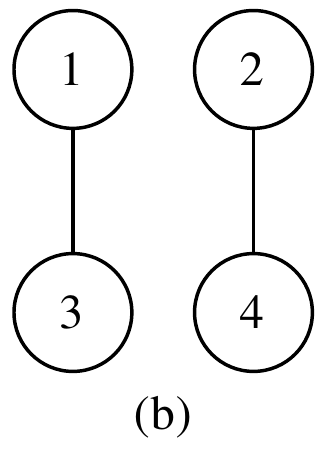}\quad\quad\quad\quad
\includegraphics[width=0.15\textwidth]{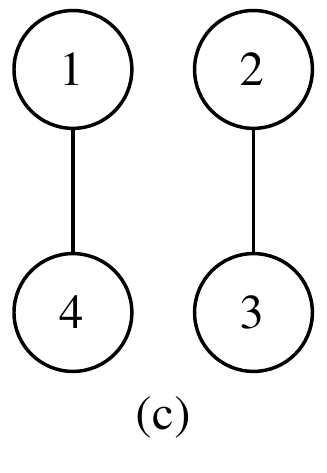}\newline
\includegraphics[width=0.33\textwidth]{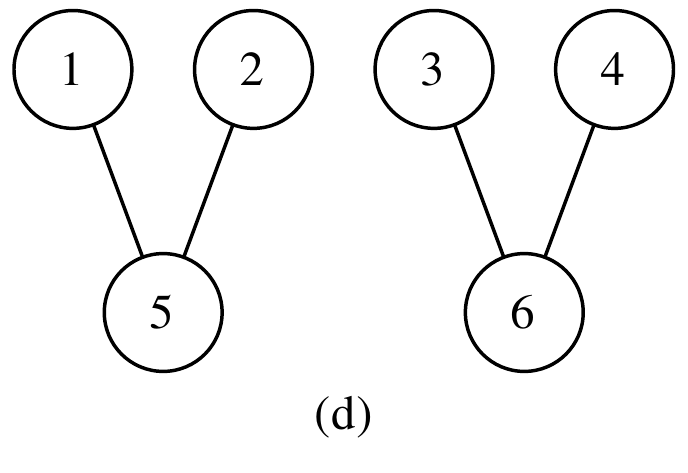}\quad\quad\quad\quad
\includegraphics[width=0.33\textwidth]{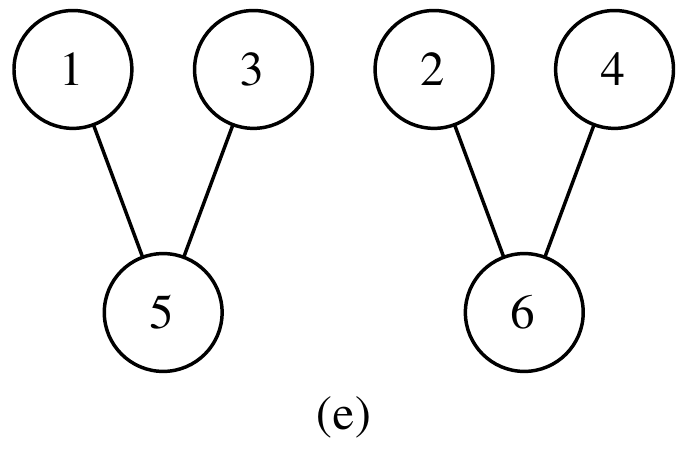}\newline
\includegraphics[width=0.33\textwidth]{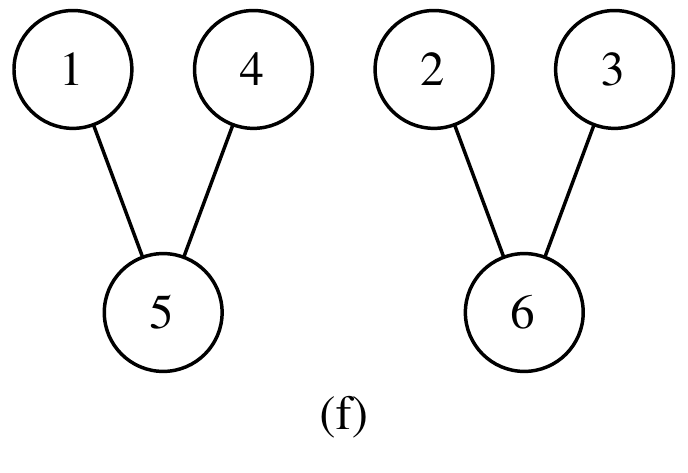}\quad\quad\quad\quad
\includegraphics[width=0.33\textwidth]{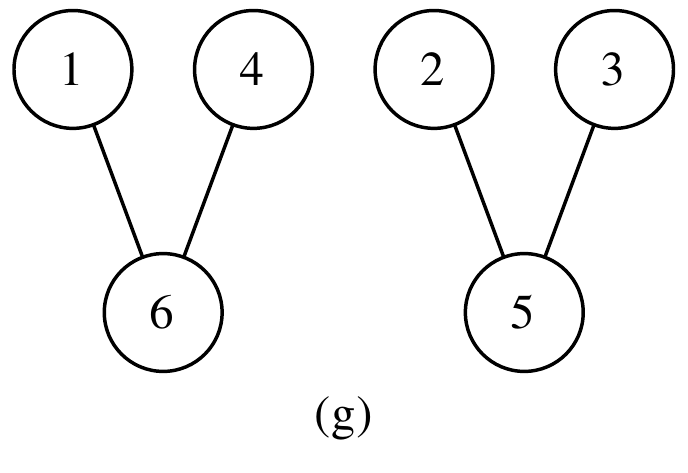}\newline
\includegraphics[width=0.33\textwidth]{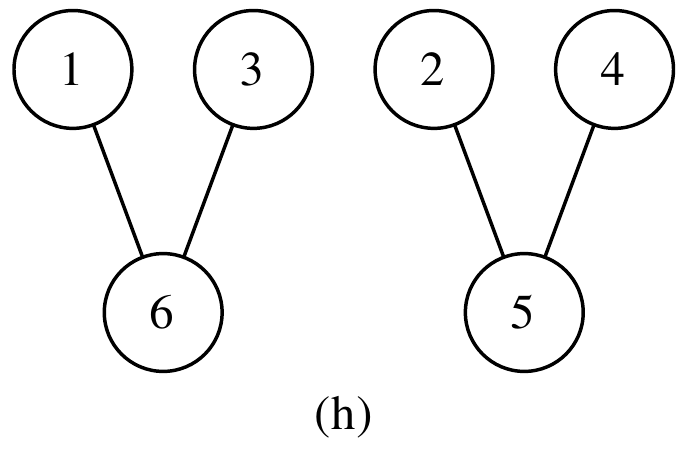}\quad\quad\quad\quad
\includegraphics[width=0.33\textwidth]{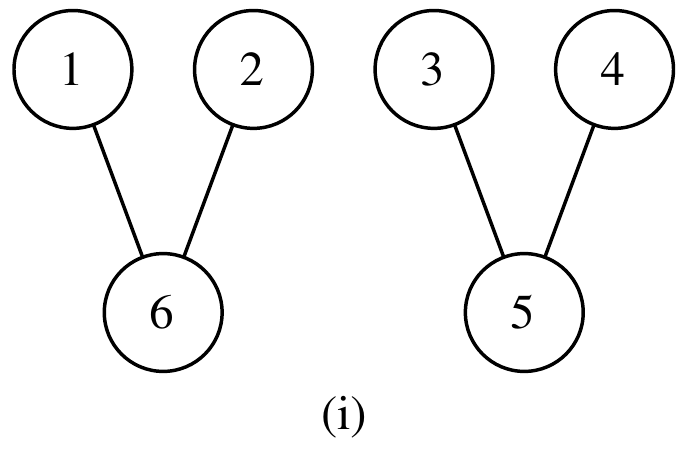}\newline
\caption{\label{type2}Degree-class subgraphs realizing the degree-spectra
matrix $S_{HB}$ of Eq.~\ref{shb}. Panels (a), (b) and (c) show the possible
realizations of $G_{2,2}$; panels (d) to (i) show the possible realizations
of $G_{2,3}$.}
\end{figure}
The next step in our evaluation is to compute the probabilities
of generating any of the hexagon and bow tie graphs from the degree-spectra
matrix $S_{HB}$. The graph-construction part of our algorithm consists
in generating all the $G_{\alpha\beta}$ subgraphs between nodes
of degree~$\alpha$ and nodes of degree~$\beta$. In the current
example, there are three such subgraphs, namely $G_{2,2}$, $G_{2,3}$,
and $G_{3,3}$. Of these, $G_{3,3}$ consists simply in a single
edge between the two nodes of degree~3. Thus, the only variability
is given by the choices for the two remaining subgraphs.

The possible realizations of $G_{2,2}$ are illustrated in panels~(a),
(b) and (c) of Fig.~\ref{type2}. Each is determined by the placement
of a single edge, which forces the choice for the remaining one. Thus,
each is produced by our algorithm with the same probability of $1/3$.
Similarly, each of the possible realizations for $G_{2,3}$, shown in
panels~(d) to~(i) of Fig.~\ref{type2}, is determined by the edges incident
to node~5 or node~6. As these are chosen by our algorithm fully randomly,
all the possible realizations occur with the same probability of $1/6$.
The particular type of graph that is produced depends on the specific
realizations of the subgraphs. As there are 3~realizations for $G_{2,2}$
and 6~for $G_{2_3}$, the total number of graphs is~18. Of these, $1/3$~are
bow tie graphs, and the remaining $2/3$ are hexagon graphs. In particular,
the bow tie graphs correspond to the subgraph choices (a,d), (a,i), (b,e),
(b,h), (c,f) and (c,g), as it is easy to see from Fig.~\ref{type2}. Note
that this indicates that, for this specific degree-spectra matrix, the
sampling is already uniform.

It is possible, now, to compute the average clustering coefficients
for the unweighted estimation. To do so, first compute their average over
the realizations of $S_{HB}$:
\begin{eqnarray}
 \la c_2\ra_{HB} &= \frac{1}{3}\cdot 1+\frac{2}{3}\cdot 0 = \frac{1}{3}\\
 \la c_3\ra_{HB} &= \frac{1}{3}\cdot\frac{1}{3}+\frac{2}{3}\cdot 0 = \frac{1}{9}\:.
\end{eqnarray}
Then, knowing that $S_{HB}$ is sampled with probability $1/27$,
and the remaining degree-spectra matrices always yield pentagon
graphs, it is
\begin{eqnarray}
 \la c_2\ra_{unweighted} &= \frac{1}{27}\cdot\frac{1}{3} + \frac{26}{27}\cdot\frac{1}{4} = \frac{41}{162}\\
 \la c_3\ra_{unweighted} &= \frac{1}{27}\cdot\frac{1}{9} + \frac{26}{27}\cdot\frac{1}{3} = \frac{79}{243}\:.
\end{eqnarray}

\subsection{Weighted estimate}
\begin{figure}
\centering
\includegraphics[width=1\textwidth]{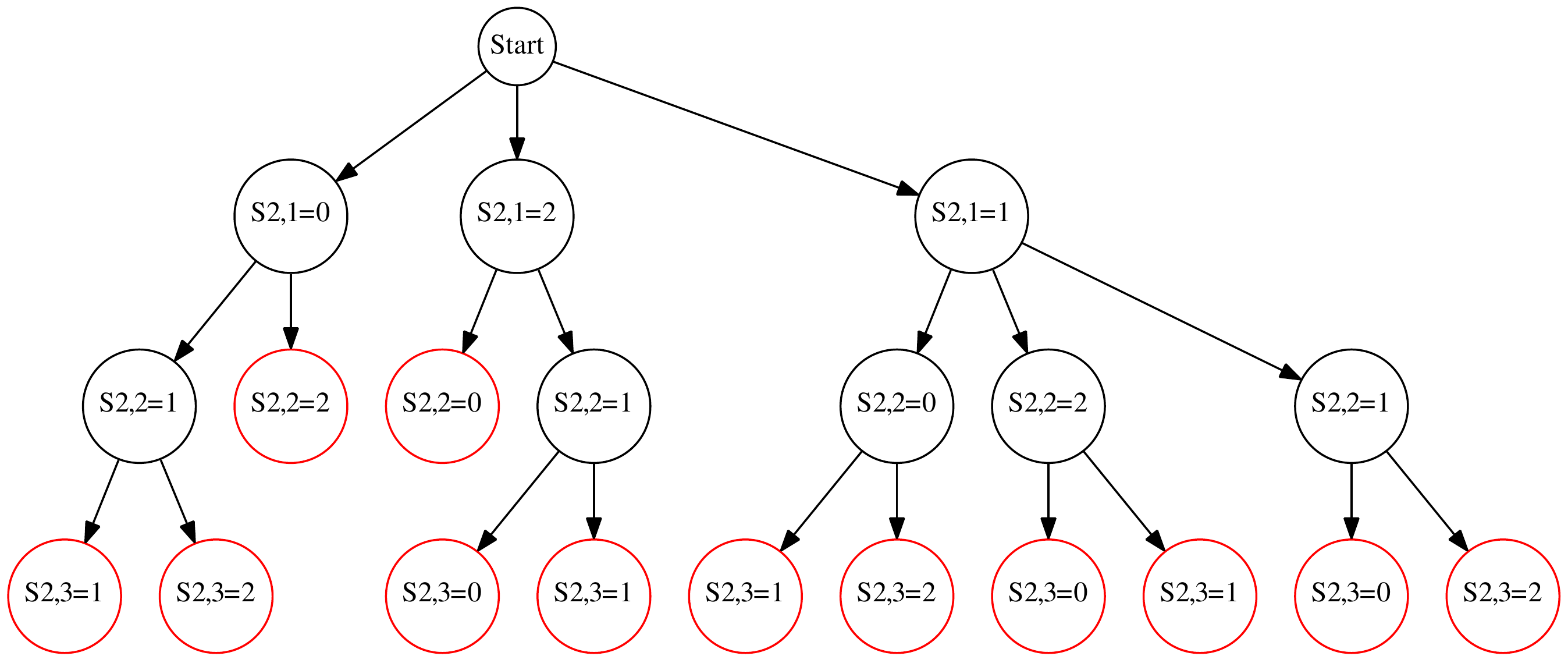}
\caption{\label{dege}Decisional tree for the construction of degree-spectra
matrices realizing $J$ and corresponding to pentagon graphs. The~12 leaves
of the tree are shown in red.}
\end{figure}
In order to obtain an analytical result for the weighted
estimate, rather than computing the probability of occurrence
of each degree-spectra matrix as sampled by our algorithm,
we need to compute their actual number. From the previous
subsection, we already know that all the hexagon and bow
tie graphs come from the same degree-spectra matrix $S_{HB}$,
which is unique. Then, we only need to compute the number
of degree-spectra matrices corresponding to pentagon graphs.

To do so, remember that the first choice in the construction
of a degree-spectra matrix from $J$ is the value of the element
$S_{2,1}$. If $S_{2,1}=0$ or $S_{2,1}=2$, then we are guaranteed
to get a pentagon graph. However, while each of these two choices
fixes the value of $S_{3,1}$, we are still free to select a
value for the next ``free'' element, $S_{2,2}$.

If $S_{2,1}=0$, the allowed values for $S_{2,2}$ are~1 and~2.
Choosing~2 fixes all the other elements of the degree-spectra
matrix. Conversely, choosing~1 results in $S_{2,3}$ still to
be determined. Its possible values are~1 and~2. Thus, there are~3
different degree-spectra matrices with $S_{2,1}=0$.

If, instead, $S_{2,1}=2$, the situation is very similar to the
first case. The possible choices for $S_{2,2}$ are~0 and~1. Choosing~0
fixes the entire matrix; choosing~1 requires to select a value
for $S_{2,3}$, which can be either~0 or~1. Thus, there are~3
matrices with $S_{2,1}=2$.

The third possibility of $S_{2,1}=1$ still allows degree-spectra
matrices corresponding to a pentagon graph. Similarly to the previous
case, the simplest way to construct one is to impose $S_{2,2}=0$
or $S_{2,2}=2$. In both cases, one must then choose a value for
$S_{2,3}$. The possibilities are~1 and~2, if $S_{2,2}=0$, or~0 and~1,
if $S_{2,2}=2$. Any choice for $S_{2,3}$ fixes all the remaining
elements of the matrix.

Finally, it is still possible to choose $S_{2,1}=1$ and $S_{2,2}=1$,
and still construct matrices corresponding to a pentagon graph. The
choice is again on $S_{2,3}$. Choosing $S_{2,3}=0$ or $S_{2,3}=2$ fixes
all the other elements of the matrix, whose realizations will be pentagon
graphs. Imposing $S_{2,3}=1$, instead results in the matrix $S_{HB}$,
exhausting all possibilities. This shows that there are~6 different
matrices with $S_{2,1}=1$ that generate pentagon graphs.

The decisional tree we just described is shown in
Fig.~\ref{dege} as a visual aid. In summary, there
are~12 different degree-spectra matrices that realize~$J$
and whose realizations are always pentagon graphs.
Knowing $\la c_2\ra_{P}$, $\la c_3\ra_{P}$, $\la c_2\ra_{HB}$
and $\la c_3\ra_{HB}$, which we computed before, we
can finally calculate the weighted average clustering
coefficients:
\begin{eqnarray}
 \la c_2\ra_{weighted} &= \frac{12}{13}\cdot\frac{1}{4} + \frac{1}{13}\cdot{1}{3} = \frac{10}{39}\\
 \la c_3\ra_{weighted} &= \frac{12}{13}\cdot\frac{1}{3} + \frac{1}{13}\cdot{1}{9} = \frac{37}{117}\:.
\end{eqnarray}

\subsection{Numerical verification}

To validate our algorithm against the analytical results
presented in the two subsections above, we performed extensive
numerical simulations, generating $10^4$ degree-spectra
matrices, and $10^4$ samples per matrix, for a total of
$10^8$ graphs. For each graph generated, we saved the average
local clustering coefficients for nodes of both degrees.
Then, we obtained both weighted and unweighted results by averaging
the data first naively, and then with a proper use of the
weights according to Eq.~\ref{degspest}. The results, shown
in Table~\ref{valid}, show that the weighted averages obtained
using our algorithm converge to the correct result. Also,
the difference between weighted and unweighted results can be
appreciated even when it is quite small, as in our example.
This illustrate the sensitivity of our method, as well as
the necessity of using proper sampling when performing
this kind of studies.
\begin{table} \label{table}
 \caption{\label{valid}Comparison between analytical and simulated
averaged local clustering coefficients.}
 \begin{tabular}{@{}lllll}
         \br
Coeff. &Theor.~unweigh. &Simul.~unweigh. &Theor.~weigh. &Simul.~weigh.\\
\mr
$c_2$ & $0.25309$ & $0.25320$ & $0.25641$ & $0.25664$\\
$c_3$ & $0.32510$ & $0.32483$ & $0.31624$ & $0.31570$\\
\br
        \end{tabular}
\end{table}

\end{document}